\documentclass[lettersize,journal]{IEEEtran}

\usepackage[paperwidth=7.875in, paperheight=10.75in, 
            top=0.75in, bottom=1in, left=0.625in, right=0.625in]{geometry}
\usepackage[T1]{fontenc}
\usepackage{mathrsfs}
\usepackage{bm}
\usepackage{amsmath}
\usepackage{breqn}
\usepackage{amsthm}
\usepackage{amssymb}
\usepackage{graphicx}
\PassOptionsToPackage{normalem}{ulem}
 \usepackage{subfig}
\usepackage{lipsum}
\usepackage{cite}
\usepackage[english]{babel}

\usepackage{algorithm}
\usepackage{algpseudocode}
\DeclareMathOperator*{\argmin}{argmin}
\usepackage{setspace}
\usepackage{xcolor}
\renewcommand{\normalsize}{\fontsize{9.5pt}{11.5pt}\selectfont}
\normalsize
\makeatletter

\theoremstyle{plain}
\newtheorem{theorem}{\protect\theoremname}
  \theoremstyle{plain}
  \newtheorem{conjecture}{\protect\conjecturename}
  \theoremstyle{plain}
  
  \theoremstyle{plain}
  \newtheorem{definition}{\protect\definitionname}
  \theoremstyle{definition}
   \newtheorem{lemma}{\protect\lemmaname}
  \theoremstyle{remark}

\theoremstyle{assumption}
    \theoremstyle{proposition}

\theoremstyle{algorithm}

\setkeys{Gin}{width=1.0\columnwidth}

\makeatother

\usepackage{babel}
  \providecommand{\definitionname}{Definition}
  \providecommand{\lemmaname}{Lemma}
  \providecommand{\propositionname}{Proposition}
  \providecommand{\remarkname}{Remark}
\providecommand{\theoremname}{Theorem}
\providecommand{\conjecturename}{Conjecture}
\providecommand{\assumptionname}{Assumption}

\begin{document}

 \title{Age-of-information minimization under energy harvesting and  non-stationary environment 
 \thanks{The authors are with the Department of Electrical Engineering, Indian Institute of Technology, Delhi (email:   akanksha.jaiswal@ee.iitd.ac.in,      arpanc@ee.iitd.ac.in.  }
 \thanks{
This paper's preliminary version was published in \cite{jaiswal2023age}.}
}

 \author{
Akanksha Jaiswal, Arpan Chattopadhyay \\
\vspace{-0.1in}
 }

\maketitle

\begin{abstract} 
This work addresses the problem of minimizing the age of information (AoI) in a wireless sensor network consisting of a sink node and  multiple  potentially mobile energy harvesting (EH) sources inducing non-stationarity in the environmemnt. At each time instant, the central scheduler selects one of the mobile sources to probe the quality of its channel to the sink node, and then the assessed channel quality is utilized to determine whether a source will collect a sample and send the packet to the sink. For the single source case, we assume that the   channel quality is known at each time instant, model the problem of AoI minimization as a Markov decision process (MDP), and prove the optimal sampling policy's threshold structure. We then use this threshold structure and propose an age-energy-channel based sliding window upper confidence bound reinforcement learning (AEC-SW-UCRL2) algorithm to handle unknown and time-varying energy harvesting rate and channel statistics, motivated by the popular SW-UCRL2 algorithm for non-stationary reinforcement learning (RL). This algorithm is applicable when an upper bound is available for the total variation of each of these quantities over a time horizon. Furthermore, in situations where these variation budgets are not accessible, we introduce a novel technique called age-energy-channel based bandit-over-reinforcement learning (AEC-BORL) algorithm, motivated by the well-known BORL algorithm. For the multiple source case, we demonstrate that the AoI minimization problem can be formulated as a constrained MDP (CMDP) which can be relaxed using a Lagrange multiplier. We decouple the problem into various sub-problems across source nodes, and also derive    Whittle’s index based source scheduling policy for probing as well as an optimal threshold policy for source sampling. We next leverage this Whittle's index and threshold structure to develop a Whittle's index and threshold based sliding window upper confidence bound reinforcement learning (WIT-SW-UCRL2) algorithm for unknown, time-varying energy harvesting rates and channel statistics under their respective variation budgets. Moreover, we also proposed a Whittle's index and threshold based bandit-over-reinforcement learning (WIT-BORL) algorithm for unknown variation budgets. Finally, we numerically demonstrate the efficacy of our algorithms.

\end{abstract}


\section{Introduction}\label{section:introduction}
For many time-sensitive sensing and communication applications such as IoT systems, vehicle tracking, environment monitoring, smart home systems, etc., the data available at the decision-making node needs to be as fresh as possible. Age of Information (AoI  \cite{kaul2012real}) is a popular measure of
the freshness of information in such systems and is studied heavily to realize ultra-reliable low-latency communication (URLLC \cite{zhang2020analyses}) in 5G systems. At any time $t$, if the decision-making node finds that the  latest available monitoring information comes from a packet whose time-stamped
generation instant was $t'$, then  the decision making node computes the AoI as $ (t-t')$.

Nevertheless, battery constraints and restricted communication bandwidth pose two significant obstacles to the implementation of such real-time status update systems. To address the scarcity of bandwidth, it is common for multiple sources to share a common frequency band in order to transmit their time-sensitive status updates to a monitoring node. On the other hand, the challenge of battery limitations can be resolved using energy harvesting source (EH) nodes. However, in order to prevent interference among these nodes and to fully leverage the EH feature available at the source nodes, it is necessary to implement intelligent scheduling policies. Furthermore, for  mobile sensors such as UAVs and mobile robots, both energy harvesting characteristics and wireless channel statistics may vary significantly with the nodes' location and trajectory, which necessitates adaptive, learning-based scheduling and sampling policies that can handle time-varying environments.

The recent efforts towards energy and bandwidth-efficient  URLLC  has led to a number of works on AoI minimization in EH systems. For instance, the papers \cite{wu2017optimal,bacinoglu2018achieving,ceran2021learning, abd2020aoi, jaiswal2021minimization, 10547087} have proposed optimal strategies to minimize AoI for a single EH source in various  network scenarios. In particular, \cite{wu2017optimal} has investigated the optimal status update policy for a EH source sending update packets through a noiseless channel, and has adopted an energy-aware, threshold-type status update policy. The paper \cite{bacinoglu2018achieving} discusses the trade-off between age and energy, and presents an optimal threshold policy for AoI minimization. The paper \cite{ceran2021learning} focuses on a threshold policy to minimize AoI for a single EH sensor using the Hybrid Automatic Repeat Request (HARQ) protocol and also provides RL algorithm for an unknown environment. A threshold-based age-optimal policy for an RF-powered source node providing status updates to the destination has been presented by the authors of \cite{abd2020aoi} where the policy considers the scenario that the source can collect energy or transfer data at a time. The AoI minimization problem was formulated as an MDP for a single EH source with channel probing capabilities in our previous paper \cite{jaiswal2021minimization, 10547087}, which demonstrated threshold structure in the optimal channel probing and source sampling strategies.

On the other hand, AoI minimization in networks with several source nodes vying for a common channel has also been extensively studied in the literature \cite{zakeri2023minimizing, wang2024scheduling, bedewy2019age,xie2021reinforcement, zhou2019joint, hatami2022demand,tang2020minimizing, abd2020reinforcement, hatami2021aoi}.  Using CMDP formulation, the authors of \cite{zakeri2023minimizing} minimized the weighted average AoI subject to transmission limitations in a multi-source relaying system with independent sources and error-prone connections. The paper \cite{bedewy2019age} presented the Maximum Age First (MAF) scheduling policy and Zero wait sampling policy to reduce the total average peak age for the optimal decoupled scheduling and sampling policy. However, the most effective sampling policy for reducing the total average age was determined to be a threshold policy. 
 The paper \cite{zhou2019joint} has proposed a semi-distributed sampling and updating
learning algorithm (via Q-functions) for multiple IoT devices under sampling/updating cost, a block-fading channel, and average energy constraint. To reduce the average AoI for a resource-constrained IoT network with a cache-enabled edge node between users and EH sensors, the authors of \cite{hatami2022demand} have proposed a relax-then-truncate optimal policy. However, the literature mostly deals with stationary systems where the statistical behaviour of the environment does not change with time. While there have been some recent efforts towards  AoI minimization in non-stationary environment \cite{banerjee2020fundamental, tripathi2021online, tang2020minimizing}.  For example, the authors of  \cite{tang2020minimizing} have decoupled the multi-sensor
scheduling problem into a single-sensor CMDP and utilized the optimal threshold policy structure of each decoupled sensor to design a scheduling algorithm for the AoI minimization problem in a multi-sensor network having time-varying channel states under bandwidth and power consumption constraints.  {\em However, no prior work has addressed AoI minimization in EH system under non-stationary environment. }

On the other hand, Whittle's Index \cite{whittle1988restless} for restless multi-armed bandits (RMABs) has been extensively employed to mitigate the computational complexity associated with resolving constrained MDPs ({\em i.e.,} CMDPs) for multi-source systems, as it offers a nearly optimal solution. Numerous papers have utilized Whittle's Index based policies to minimize AoI in different multi-source scheduling systems \cite{kadota2016minimizing, sun2019closed, hsu2018age, tang2021whittle, tong2022age, kriouile2021global, tripathi2019whittle}, but without taking into account energy harvesting sources and channel probing capabilities at the sources. For example, the paper \cite{kadota2016minimizing} examined a broadcast network with a single base station (BS) and many clients, and demonstrated that the optimal action in the symmetric network is to transmit a packet with the highest age in a greedy manner. For general network scenario, the authors proved that the problem is indexable and derived Whittle's index in closed form. The authors in 
 \cite{sun2019closed}   derived closed-form expressions for Whittle's indices and proposed an Index-Prioritized Random Access (IPRA) scheme to minimize AoI with multiple sources and a single sink setting.
 The average age minimization problem for a system model in which only a fraction of users can concurrently transmit packets to a BS over unreliable channels and the BS can decode these packets with some success probabilities was investigated in the paper \cite{kriouile2021global}. It employed a Cauchy criterion to demonstrate the optimality of Whittle's index policy.
The paper \cite{tong2022age} has proposed a generalized Whittle index (GWI)
and the generalized partial Whittle index (GPWI) scheduling policies to solve the decoupled constrained RMAB (CRMAB) problem formulated for AoI minimization in IoT networks having correlated sources with identical and nonidentical multi-channel settings, respectively.

In order to solve RMABs, it is essential to have information on transition dynamics, which are often unknown in advance. Several online learning methods to enhance planning efficiency in situations involving unknown transitions in RMAB have been presented in the literature \cite{ceran2021reinforcement, wang2023optimistic, mao2020model, akbarzadeh2023learning, xiong2022learning, fu2019towards, DBLP:journals/corr/abs-2004-14427, killian2021q}. One example is the UCWhittle algorithm \cite{wang2023optimistic} which is based on Whittle's index policy and it uses the upper confidence bound to learn transition probabilities. Additionally, the paper \cite{akbarzadeh2023learning} has proposed Thompson-sampling based learning algorithm for a restless bandit, \cite{xiong2022learning} introduced two index-aware RL algorithms, GM-R2MAB and UC-R2MAB for infinite-horizon average-reward RMABs and for the unknown channel statistics. Different RL methods, like UCRL2, DQN, and average-cost SARSA with LFA have been looked at and compared numerically in \cite{ceran2021reinforcement}.

The authors in \cite{auer2008near} have proposed upper confidence bound for reinforcement learning (UCRL2) algorithm for the stationary setting. Also, there are several papers that have studied non-stationary RL in MDPs and provide their regret bound \cite{ gajane2018sliding, li2019online, abdallah2016addressing,padakandla2020reinforcement, ding2022provably, cheung2023nonstationary}.  The authors of  \cite{gajane2018sliding} have developed sliding window upper-confidence bound for reinforcement learning (SW-UCRL) algorithm for non-stationary MDP with arbitrarily changing rewards and state transition distributions.  In \cite{cheung2023nonstationary},  undiscounted RL  has been considered under a  setting where both the reward and state transition probabilities vary with time under total variation budget constraints,  and provide dynamic regret bound. The authors of \cite{cheung2023nonstationary} have developed the SW-UCRL2 algorithm with confidence widening (SW-UCRL2-CW) and known budgets for its temporal variations. They have also proposed the bandit-over-reinforcement learning (BORL) algorithm which tunes system parameters (sliding window size and a confidence-widening parameter) for running SW-UCRL2-CW under unknown variation budgets. Consequently, people have started applying bandit algorithms for AoI minimization. For example, the authors in \cite{9559999} have considered the AoI bandit setting, where, a single source measures a time-varying process and schedules the status updates on one of the available error-prone channels. They have proposed AoI aware upper confidence bound (UCB \cite{auer2002finite}) and Thompson sampling \cite{russo2018tutorial,thompson1933likelihood} policies which outperform the existing age unaware policies.

 However, there is no paper that has considered optimal source scheduling and sampling problem to minimize average AoI for multiple EH source systems with a shared lossy channel under a non-stationary environment, which is the subject matter of our paper.
 In this work, we address this important problem assuming that the energy harvesting process and the channel statistics exhibit temporal non-stationarity. However, as a precursor to the problem under a non-stationary environment, we first consider a stationary environment for single source case with known, time-invariant parameters characterizing energy harvesting statistics and channel statistics, formulate the problem of AoI minimization as an MDP, and analytically establish the threshold structure of the optimal policy. Next, we consider the non-stationary environment where the mentioned parameters are unknown and time-varying, and formulate the AoI minimization problem as a non-stationary RL problem. While there have been popular algorithms such as SW-UCRL \cite{gajane2018sliding}  and BORL \cite{cheung2023nonstationary} for the non-stationary RL problem, we leverage the threshold nature of the optimal policy under the stationary setting, and propose two algorithms called AEC-SW-UCRL2 (when there is a known upper limit on the variation of the environment) and AEC-BORL (when there is no known temporal variation budget for the environment); each of these two algorithms has a critical threshold comparison step along with other steps, and the choice of the threshold values are motivated by age-aware  Thompson sampling \cite{9559999}. For the multiple source case, we first recall the results obtained from our previous work \cite{jaiswal2024whittlesindexbasedageofinformationminimization} which formulates the AoI Minimization problem for multiple sources under a stationary environment as a CMDP,   uses the Lagrangian relaxation method to relax it, and further decouples it into subproblems for various sources. For decoupled subproblems, our previous work \cite{jaiswal2024whittlesindexbasedageofinformationminimization} proposed a  Whittle’s index and threshold based source scheduling and sampling policy (WITS3 policy) which, at each time instant, schedules a source with the highest Whittle's index to probe the channel state. For the selected source, our previous work \cite{jaiswal2024whittlesindexbasedageofinformationminimization} shows that it is optimal to sample a source if the probed channel quality is above a threshold. Following this, for the nonstationary RL problem in a multi-source system, we use Whittle's index and threshold policy to propose a WIT-SW-UCRL2 algorithm for known variation budget and WIT-BORL algorithm for unknown variation budget.  Numerical results clearly show that AEC-SW-UCRL2,  AEC-BORL, WIT-SW-UCRL2, and WIT-BORL outperform their existing counterparts in terms of cumulative AoI under an unknown, non-stationary environment.



\subsection{Our contributions  }
\begin{enumerate}
\item In Section~\ref{section:single-sensor-single-process_stationary}, for the single source case, we study AoI minimization problem under a stationary environment. We prove that the optimal policy amounts to checking whether the packet success probability under the current channel state exceeds a threshold. 
\item In Section~\ref{section:Age Aware Sliding Window UCRL2}, we consider the AoI minimization problem for a single mobile source under time-varying energy harvesting characteristics and channel statistics,  formulate the problem as a non-stationary RL problem, and propose the AEC-SW-UCRL2 algorithm that exploits the known variation budgets on the energy generation rate and channel statistics.  Our algorithm is different from the standard SW-UCRL algorithm \cite{gajane2018sliding}  because we, motivated by the threshold policy structure in Section~\ref{section:single-sensor-single-process_stationary}, introduce a thresholding step on the age. The threshold depends on the available energy level, the current channel condition, and the energy generation history.  
\item In Section~\ref{section:Age Aware BORL}, we relax the assumption on the availability of the variation budgets, and motivated by the BORL algorithm  \cite{cheung2023nonstationary}, obtain AEC-BORL by incorporating the same threshold comparison step as in Section~\ref{section:Age Aware Sliding Window UCRL2}. 
\item In Section~\ref{section: WIT Sliding Window UCRL2}, we formulate a non-stationary RL AoI minimization problem for multiple EH source system with mobility-driven time-varying energy harvesting characteristics and channel statistics, and propose a novel algorithm called WIT-SW-UCRL2 for known variation budgets. Motivated by the Whittle's index and threshold based source scheduling and sampling policy for a stationary multi-source system, our algorithm selects the source with the highest Whittle's index for channel probing. Sampling decision on the probed source is taken if the age of that source exceeds a threshold. 
\item In Section~\ref{section: WIT BORL}, we propose a WIT-BORL algorithm that does not require the assumption about the availability of variation budgets, while retaining the Whittle's index and threshold comparison steps as described in Section~\ref{section:Age Aware Sliding Window UCRL2}. The WIT-BORL algorithm tunes window size for different blocks of the entire time horizon and calls a WIT-SWUCRL2 algorithm as a subroutine with a selected window size to choose the action for that particular block. 

\end{enumerate}

\subsection{ Organization }

The subsequent sections of the paper are organized as follows. Section~\ref{section:system-model} describes the system model. Section~\ref{section:single source non-stationary-pf} focuses on AoI minimization for a single source working in a non-stationary environment, whereas Section~\ref{section: multi-source-nonstat} explores AoI minimization for multiple sources in a non-stationary environment under  source scheduling constraint. Lastly, simulation results are presented in Section~\ref{section: numerical_results}, followed by the conclusions  in Section~\ref{section:Conclusion}. All proofs are included in the appendices.

{\bf Notation:} In this paper, $\mathbb{P}(\cdot)$,
 $\mathbb{E}(\cdot)$,  and $\mathbf{1}$ represent the probability operator, expectation operator, and indicator term, respectively. 
\section{System model}\label{section:system-model}

We consider a wireless network system consisting of $N$ mobile source nodes, a sink node, and a central scheduler/decision maker (DM) as in Figure~\ref{multis-system-model-ns}.  At each discrete time instant $t$, the scheduler selects one of the sources $i(t) \in \{1,2,\cdots, N\}$ for probing the state of its  channel to the sink. Next,  the selected source $i(t)$ further decides on sampling and transmission of the status update to the sink node. Each source~$i$ harvests energy from the environment and stores it in its finite buffer of size $B_i$ units. At time~$t$, let us denote by  $A_{i}(t)$ and $E_{i}(t) \in \{0,1,\cdots, B_{i}\}$ the number of energy packet arrivals, and the available energy, respectively for the $i$-source node. Here, $\{A_{i}(t)\}_{t \geq 0}$ is considered to be an i.i.d.  Bernoulli process with unknown mean $\lambda_{i,t}>0$; i.e., $\mathbb{E}(A_{i}(t))=\lambda_{i,t}$. We assume that the resource spent in channel probing is negligible whereas sampling decisions cost the $i$-th source $E_s$ units of energy.  To address the unreliable nature of wireless networks, we consider a fading channel between the $i$-th source and the sink node. At time $t$,     the channel state between the $i$-th source node and the sink node is denoted by $C(i, t)$. We assume that $C(i,\cdot)  \in \{C_1, C_2,\cdots, C_m\}$, where $m$ is a finite number representing the total number of possible channel states. 

\begin{figure}[t]
  \begin{center}
 \includegraphics[height=3.8cm,width=7cm]{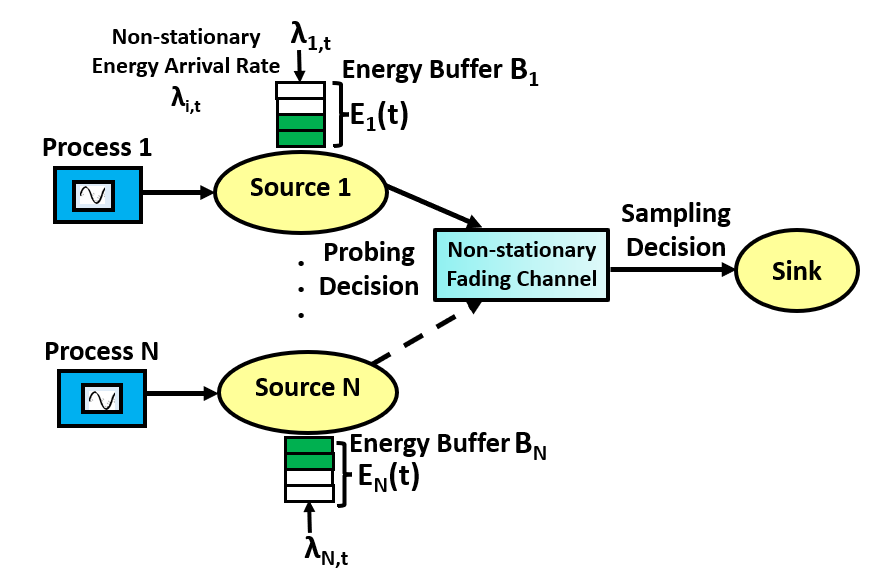}
 \caption{Remote sensing system with multiple EH sources in non-stationary environment.}
 \label{multis-system-model-ns}
 \end{center}
 \vspace{-5mm}
\end{figure}

Additionally, $p(i,t)$ represents the probability that the packet is successfully transmitted from the $i$-th source to the sink node at time $t$, where $p(i,t) \in \{p_1, p_2, \cdots, p_m\}, \forall i \in \{1,2,\cdots,N\}$.  The term $r_{i}(t)$ is a binary indicator of successful packet transmission from the $i$-th source to the sink at the time $t$.
Hence,  $\mathbb{P}(r_{i}(t)=1|C(i, t)=C_j)=p_j$ and similarly, $\mathbb{P}(p(i,t)=p_j|C(i, t)=C_j)=1$ for all $j \in \{1,2,\cdots,m\}$ and $i \in \{1,2,\cdots, N\}$. 
Here, we assume that the channel states across distinct nodes are independent and that $\{C(i, t)\}_{t \geq 0}$ is i.i.d. across $t$ for each $1 \leq i \leq N$. Let us  define $q_{j,i,t}$ as the probability that the channel state of source $i$ to the sink node is $C_j$ at time $t$, i.e.  $q_{j,i,t} \doteq \mathbb{P}(C(i, t)=C_j)$, where $1 \leq j \leq m$.  Each source node~$i$ is assumed to be mobile, leading to time-varying energy harvesting rates $\lambda_{i,t}$ and channel state probabilities $q_{j,i,t}$, and consequently  non-stationarity in both the energy arrival process  and the channel state distribution.

It is assumed that only one source-sink channel out of $N$ is probed to determine its state $C(i,t)$ at each time~$t$.  
 The indicator of probing the channel state of source $i$  is denoted by $b_{i}(t) \in \{0,1\}$. To prevent pilot  collision, only one source is selected to estimate the channel, and hence $\sum_{i=1}^N  b_{i}(t) \leq 1 ,   \forall t \geq 1$. 
 Furthermore, for   source $i$ such that $b_i(t)=1$, we define $a_{i}(t) \in \{0,1\}$ as the indicator that represents the decision to either sample and transmit the data packet from the $i$-th source to the sink node ($a_{i}(t)=1$), or to remain inactive ($a_{i}(t)=0$). Therefore,   if $b_{i}(t)=0$, then $a_{i}(t)=0$. Conversely, if $b_i(t)=1$, then $a_i(t) \in \{0,1\}$. 
A generic action for the $i$-th source at time $t$ is represented as $(b_{i}(t),a_{i}(t))$, and the action space of source $i$ is therefore given by $\mathcal{A}_{i}=\{(0,0), (1,0), (1,1)\}$. It is important to note that $a_i(t)=1$ is possible only if $E_i(t) \geq E_s$.

Let  $\kappa_{i}(t)\doteq \sup\{0 \leq \kappa < t: a_{i}(\kappa)=1, r_{i}(\kappa)=1\}$ be defined as the last time instant before time~$t$ when the $i$-th source had successfully sent a status update to the sink node. The AoI for the $i$-th source at time~$t$ is  denoted by $K_i(t)=(t-\kappa_i(t))$. However, when $a_{i}(t)=1$ and $r_{i}(t)=1$, $K_i(t)=0$ since the sink node has access to the most recent status update of the $i$-th source.

Note that, for the single source case ($N=1$), we omit the node index $i$ from the above-defined variables (for example $A_i(t)= A(t), E_i(t)=E(t), C(i,t)= C(t), p(i,t)=p(t), \lambda_{i,t}= \lambda_{t}, q_{j,i,t}= q_{j,t}$ and so on). Since there is no cost of channel probing, $b(t)=1, \forall t \geq 1$, and hence the action space in this case reduces to  $\mathcal{A}=  \{0,1\}$ containing the sampling decision indicator only.

\section{Single source $(N=1)$}\label{section:single source non-stationary-pf}
In this section, we first develop a threshold decision rule for the stationary setting, which is later used in the sampling and communication algorithm design for the non-stationary setting.

\subsection{MDP Formulation for Stationary  Setting } \label{section:single-sensor-single-process_stationary}
In this section, we assume that $\lambda_t=\lambda$ and $q_{j,t}=q_j$ for all $j \in \{1,2,\cdots,m\}$ and for all $t \geq 0$. We formulate the problem of AoI minimization as an  MDP with state space $\mathcal{S} \doteq \{0,1,\cdots, B\} \times \mathbb{Z}_+  \times \{C_{1}, C_{2},\cdots, C_{m}\}$; a generic state $s = (E, K, C)$ denotes that  $E$ energy packets are available in the energy buffer, the last successfully received packet was generated at the source node $K$ slots earlier, and the current channel state is  $C$.  The action space is $\mathcal{A}=  \{0,1\}$, and the action at time~$t$ is $a(t)$ as discussed in Section~\ref{section:system-model}.  At any time, if the source node decides not to sample the process,  then the single-stage AoI cost is $c(s, a=0) = K$, and if it decides to sample, then the expected single-stage AoI cost is  $c(s,a=1) = K(1-p(C))$ which accounts for the random packet loss with probability $p(C)$. The time-averaged AoI minimization problem can be solved by solving an $\alpha$-discounted cost stationary MDP problem \cite[Section 4.1]{bertsekas2011dynamic} and then taking $\alpha \uparrow 1$. The $\alpha$-discounted cost problem is given by:
\begin{eqnarray}\label{eqn:discounted-cost-problem}
\min_{\mu} \lim_{M \rightarrow \infty} \sum_{t=0}^{M-1} \alpha^t \mathbb{E}_{\mu}(c(s(t), a(t))) 
\end{eqnarray}

where $\mu$ is a generic stationary deterministic policy. Denoting by $J^*(\cdot)$ the optimal value function, we write its  Bellman equation as follows:
\begin{eqnarray} \label{eqn:Bellman-eqn-SSSP_Probingt}
&&J^{*}(E\geq E_{s},K,C)\nonumber\\
&=&min\{K+\alpha  \mathbb{E}_{A,C^{'}}J^{*}(min\{E+A,B\},K+1, C^{'}),\nonumber\\ 
&&K(1-p(C))+\alpha p(C)\mathbb{E}_{A,C^{'}}J^{*}(min\{E-E_{s}+A,B\},\nonumber\\
&&1,C^{'})+\alpha (1-p(C))\mathbb{E}_{A,C^{'}}J^{*}(min\{E-E_{s}+A,B\},\nonumber\\
&&K+1,C^{'})\}\nonumber\\
&&J^{*}(E<E_{s},K, C)\nonumber\\
&=&K+\alpha \mathbb{E}_{A,C^{'}}J^{*}(min\{E+A,B\},K+1,C^{'}) 
\end{eqnarray}
\normalsize

Here $\mathbb{E}(A)=\lambda$, and $C'$ is viewed as the random next channel state with $\mathbb{P}(C'=C_j)=q_j, 1 \leq j \leq m$.  
The first term in the minimization operator of the first equation in \eqref{eqn:Bellman-eqn-SSSP_Probingt} is the cost of not sampling a process ($a(t)=0$), which includes single-stage AoI cost $K$  and an $\alpha$ discounted future cost which is averaged over the distributions of the number of energy packet generation $A$ and the next channel state $C^{'}$. The second term in the minimization operator of the first equation in \eqref{eqn:Bellman-eqn-SSSP_Probingt} is the cost of sampling a process ($a(t)=1$)  which includes the expected AoI cost $K(1-p(C))$ (expectation is taken over the packet success probability $p(C)$), and the next random state is given by $(min\{E-E_s+A,B\},1,C^{'} )$ and $(min\{E-E_s+A,B\},K+1,C^{'})$ if $r(t)=1$ and $r(t)=0$, respectively. The second equation in \eqref{eqn:Bellman-eqn-SSSP_Probingt} follows similarly since when $E<E_s$, the only possible action is to choose $a(t)=0$.

\subsubsection{Policy structure}

\begin{lemma} \label{lemma:SSSP-J-increasing-in-K}
 The value function 
$J^{*}(E,K,C )$  increases in $K$.  
\end{lemma}
\begin{proof} See Appendix ~\ref{appendix:proof-of-lemma-SSSP-J-increasing-in-K}.
\end{proof}
\begin{theorem}\label{theorem:SSSP-policy-p}
 At any time $t$, the optimal policy exhibits a threshold structure on $p(C)$. For any $E \geq E_{s} $, it is optimal to sample the source node if  $p(C)\geq p_{th}(E,K)$ for a suitable threshold function $p_{th}(E,K)$ of $E$ and $K$.
\end{theorem}
\begin{proof} See Appendix ~\ref{appendix:proof-of-theorem-SSSP-policy-p}.
\end{proof}
Now we make a conjecture to prove that the optimal policy also involves checking whether $K$ exceeds a threshold.
\begin{conjecture}\label{conjecture:SSSP-with-fading-policy-structure}
  The  difference   $J^{*}(min\{E+A,B\},K,C^' )-J^{*}(min\{E-E_{s}+A,B\},K,C^')$ is increasing in $K$. 
\end{conjecture}

\begin{theorem}\label{theorem:SSSP-policy-K}  Under Conjecture~\ref{conjecture:SSSP-with-fading-policy-structure}, the optimal  policy for \eqref{eqn:discounted-cost-problem} is a threshold policy on $K$. For any $E \geq E_{s} $, it is optimal to sample the source if and only if  $K\geq K_{th}(E, C)$ for a suitable threshold function $K_{th}(E,C)$. 
\end{theorem}
\begin{proof} See Appendix ~\ref{appendix:proof-of-theorem-SSSP-policy-K}.
\end{proof}

Conjecture \ref{conjecture:SSSP-with-fading-policy-structure} has been validated numerically later in Section \ref{subsection:numerical result Stationary system model} where Figure \ref{fig_Thres}(a) shows a threshold policy structure on age for sampling a source. 

\subsection{AEC-SW-UCRL2 algorithm for non-stationary environment}\label{section:Age Aware Sliding Window UCRL2} 
Motivated by the threshold policy structure for stationary MDP, we propose an algorithm called AEC-SW-UCRL2 for a non-stationary setting where $\lambda_t$ and $\{q_{j,t}\}_{1 \leq j \leq m}$ vary with $t$. The corresponding MDP is described by the tuple $(\mathcal{S}, \mathcal{A}, T, c, w)$, where  state space $\mathcal{S} $, action space $\mathcal{A}= \{0,1\}$ and  single stage cost $c(\cdot, \cdot )$ are as in Section~\ref{section:single-sensor-single-process_stationary}.   The number $T$ denotes the (finite) time horizon length.  The collection of state transition probabilities for all time instants is denoted by $w = \{w_t\}_{t=1}^T$, where  $w_t(\cdot|s, a)$  is the probability distribution of the next state given that the state-action pair at time~$t$ is $(s, a)$. 

We also define the following  variation budgets $V_\lambda$ and $V_{q}$:
\begin{align}\label{eqn:main-problem1}
V_\lambda \doteq  \sum_{t=1}^{T-1} V_{\lambda,t},  \hspace{1cm}   V_{q} \doteq \sum_{t=1}^{T-1} V_{q, t} 
\end{align}
where $V_{\lambda,t} \doteq  |\lambda_{t+1} - \lambda_t|$ and
$V_{q, t} \doteq max_{1 \leq j \leq m} |q_{j,t+1} - q_{j,t}|$. Here $V_{\lambda,t}$ and $V_{q,t}$ are unknown to the learner. However, in this section, we assume that $V_\lambda$ and $V_q$ are known and are used for selecting sliding window size $W$, though   this assumption will be relaxed in the next section.

Let us also define $V_{w,t} \doteq max_{s \in \mathcal{S}, a \in \mathcal{A}}\lVert {w}_{t+1}(\cdot|s,a)- {w_t}(\cdot|s,a)\rVert_{1}$ and variation budget on state transition probabilities as  $V_{w}\doteq\sum_{t=0}^{T-2} V_{w, t}$. The following lemma relates $V_w$ with $V_\lambda$ and $V_q$. 
\begin{lemma} \label{lemma:var_boudget-relation}
For each $t$,  we have 
$V_{w,t}\leq V_{\lambda,t+1}+V_{q,t+1}$, and consequently  $V_{w}\leq V_{\lambda}+V_{q}$.
\end{lemma}

\begin{proof}
See Appendix ~\ref{appendix:proof-of-var_boudget-relation}.
\end{proof}

\vspace{-5pt}
{\footnote{Alternatively, the goal is to  minimize   the dynamic regret \cite{fei2020dynamic};
$\min_{\pi} \underbrace{\sum_{t=1}^T ( \mathbb{E}_{\pi}[c(s(t), a(t))]-\zeta^*(t) )}_{\doteq R_T(\pi)},$ where $\zeta^*(t)$ is the optimal long-term average cost of a stationary MDP with energy arrival rate and the channel state distribution being fixed at $\lambda_t$ and  $\bm{q}_t\doteq [q_{1,t}, q_{2,t},\cdots, q_{m,t}]$ for all time instants. }} The goal of the source is to minimize the cumulative expected cost over all (not necessarily stationary or deterministic) policies $\pi$:
\begin{eqnarray}\label{eqn:main-problem}
\mathbb{E}_{\pi}[\sum_{t=1}^T c(s(t), a(t))]
\end{eqnarray}

\subsubsection{Algorithm description}
AEC-SW-UCRL2  given in Algorithm~\ref{alg:cap} uses a sliding window of size $W$ to estimate the energy arrival rate and state transition probabilities and find their confidence regions, by only time averaging at most $W$ number of recent samples. 
The algorithm divides the entire time horizon $T$ into a sequence of episodes. The start and end times of any episode $\rho$ are given by $\tau(\rho)$ and  $\tau(\rho+1)-1$,  with $\tau(1)=1$.

For any time step $t$ in episode $\rho$, we use the empirical mean estimator of $\lambda_t$ as follows: 
\begin{eqnarray} \label{eqn:lambda_estimate}
\hat{\lambda}_{t}&=& \frac{e_t}{n_t} 
\end{eqnarray}
where  $e_t \doteq \sum_{h=max\{(\tau(\rho)-W),1\}}^{t} A(h)$ and $n_t \doteq t-max\{(\tau(\rho)-W),1\}+ 1$. 
Also, for each state-action pair $(s, a)$ and each time $t$ in episode $\rho$, we define the following counters:
\begin{eqnarray} \label{eqn:state-action-count}
N_{t}(s,a)&=& \sum_{h=max\{(\tau(\rho)-W),1\}}^{t-1} \mathbf{1}\{(s(h),a(h))=(s,a)\}\nonumber\\
N_{t}^+(s,a)&=& max\{N_{t}(s,a),1\}
\end{eqnarray}
The number of times we encounter  $(s(t), a(t))=(s, a)$ in episode $\rho$ up to time $t$   as:
\begin{eqnarray} \label{eqn:state-action-count-within-episode}
f_{\tau(\rho)}(s,a)&=& \sum_{h=\tau(\rho)}^{t-1} \mathbf{1}\{(s(h),a(h))=(s,a)\}
\end{eqnarray}

The DM ends the $\rho^{th}$ episode  either if the time index $t$ is a multiple of $W$, or if $f_{\tau(\rho)}(s,a)$ $\geq$  $N^{+}_{\tau(\rho)}(s,a)$. This   condition to end an episode is known as the doubling criterion, which makes sure that each episode is increasingly large enough to facilitate learning. 

We use  an empirical mean estimator for the transition probabilities as: 
\begin{eqnarray} \label{eqn:dist_estimate}
&&\hat{w}_{t}(s'|s,a)\nonumber\\
&=& \frac{1}{N_{t}(s,a)^{+}} \bigg(\sum_{h=max\{(\tau(\rho)-W),1\}}^{t-1} \mathbf{1}\{s(h)=s, a(h)=a, \nonumber\\ 
&&s(h+1)=s'\}\bigg) 
\end{eqnarray}
\normalsize

and its confidence region as: 
\begin{eqnarray} \label{eqn:confidence-region}
&&\mathcal{H}_{w,t}(s,a) \nonumber\\
&\doteq& \bigg\{{w}:\lVert {w}(\cdot|s,a)- \hat{w}_{t}(\cdot|s,a)\rVert_{1} \leq \sqrt{\frac{14S \log(2A\tau(\rho) / \delta)}{N_{\tau(\rho)}(s,a)^{+}}} \bigg\} \nonumber\\ 
\end{eqnarray}
\normalsize

with confidence parameter $\delta \in (0,1)$, $S= |\mathcal{S}|$, $A= |\mathcal{A}|$  and confidence radius $confr_{w,\tau(\rho)}(s,a)=\sqrt{\frac{14S \log(2A\tau(\rho) / \delta)}{N_{\tau(\rho)}(s,a)^{+}}}$.\\

In the AEC-SW-UCRL2 algorithm, the DM follows a stationary policy $\Bar{\pi}_{\tau(\rho)}$ throughout the episode $\rho$ and the optimal decision is computed based on the age, energy availability, and channel quality. At any time $t$ in episode $\rho$, the AEC-SW-UCRL2 algorithm first checks whether the age $K(t)$  is greater than a threshold   $Thres(t)$ (which is a function of energy and current channel success probability); if it is true,   then  $a(t)=1$ is chosen provided that $E(t) \geq E_s$. Otherwise, if $K(t) < Thres(t)$, then the optimal policy $\tilde{\pi}_{\tau(\rho)}$  is derived by using extended value iteration (EVI \cite{auer2008near}) provided that $E(t) \geq E_s$,  which uses the   prior history   from $\max\{\tau(\rho)-W,  1\}$ to  $(\tau(\rho)-1)$ to estimate the confidence regions $\mathcal{H}_{w, \tau(\rho)}(s,a)$ for all $(s,a)$ pairs. 

The choice of $Thres(t)$ in AEC-SW-UCRL2 is motivated by the intuition that the threshold should decrease with available energy, energy harvesting rate, and packet success probability in a stationary setting. 
The current channel success probability  $p(t) \in \{p_1, p_2,\cdots, p_m\}$,  $\mathcal{H}_{w, \tau(\rho)}(s,a)$ for all $(s,a)$ pairs,  and precision parameter $\epsilon$ (which assumes value  $1 / \sqrt{ \tau(\rho)}$) are used as input parameters for EVI (see Algorithm~\ref{alg:EVI_cap}) which yields the stationary policy $\tilde{\pi}_{\tau(\rho)}$.



\begin{algorithm}
\caption{Age-Energy-Channel based Sliding Window  UCRL2 Algorithm }\label{alg:cap}
\begin{algorithmic}
\State $\bm{Input}: \mathcal{S},\mathcal{A}, T, W$.
\State $\bm{Initialize}$: $\hat{\lambda}_{0}=0$,  $t\leftarrow 1$,  initial state   $s(1)$.
\For{ $\text{episodes}$ $ \rho=1, 2,3 \cdots$} 
\State $\bm{Set}$: $\tau(\rho) \leftarrow t$, $f_{\tau(\rho)}(s,a)  \leftarrow 0$, $\forall (s,a)$. 
\State $\bm{Compute}$: $\hat{\lambda}_{\tau(\rho) }$ and $N_{\tau(\rho)}(s,a)^{+}$ for all $(s,a)$ by using \eqref{eqn:lambda_estimate} and \eqref{eqn:state-action-count}, respectively. 
\State $\bm{Compute}$: $\alpha(\tau(\rho))= e_{\tau(\rho)-1}+1,$ 
\State $\beta(\tau(\rho))= n_{\tau(\rho)-1 }-e_{\tau(\rho)-1} +1 $.
\While{ $t$ is not multiple of W and $f_{\tau(\rho)}(s(t),\Bar{\pi}_v(s(t))) \leq  N_{\tau(\rho)}^{+}(s(t),\Bar{\pi}_v(s(t)))$}
\If{ $E(t) \geq E_s$}\\
Observe the current channel state $C(t)$ and consequently $p(t)$.
\State $ \bm{Let}$: $Thres(t)$  $= max\big(\frac{\alpha(\tau(\rho))+\beta(\tau(\rho)}{E(t)\alpha(\tau(\rho))}, 1/p(t)\big) $, 
\If{ Age $K(t) \geq Thres(t) $} \\
 Choose $a(t)=1$,  observe age cost $c(s(t), a(t))$ and the next state $s(t+1)$.\\
Update: $f_{\tau(\rho)}(s(t), a(t))  \leftarrow f_{\tau(\rho)}(s(t), a(t))+1$, $ t \leftarrow t+1$.
\Else{  Find  a $ (1 / \sqrt{\tau(\rho)})$ optimal policy $\tilde{\pi}_{\tau(\rho)}$ by using EVI algorithm with
 confidence region $\mathcal{H}_{w,\tau(\rho)}(s,a)$, $\forall$ $(s,a)$ computed by \eqref{eqn:confidence-region}: \\
EVI$(\mathcal{H}_{w,\tau(\rho)}, p(t);1 / \sqrt{ \tau(\rho)})$ $\rightarrow$ $(\tilde{\pi}_{\tau(\rho)})$. \\
Choose $a(t)=\tilde{\pi}_{\tau(\rho)}(s(t))$,  observe age cost $c(s(t), a(t))$ and the next state $s(t+1)$}.\\ 
Update: $f_{\tau(\rho)}(s(t), a(t))  \leftarrow f_{\tau(\rho)}(s(t), a(t))+1$, $ t \leftarrow t+1$.
\EndIf
\Else{ $a(t)=0$,  observe age cost $c(s(t), a(t))$ and $s(t+1)$}.
Update: $f_{\tau(\rho)}(s(t), a(t))  \leftarrow f_{\tau(\rho)}(s(t), a(t))+1$, $ t \leftarrow t+1$.
\EndIf
\If{$t>T$}
\\The algorithm is terminated. 
\EndIf
\EndWhile
\EndFor
\end{algorithmic}
\end{algorithm}

\begin{algorithm}
\caption{ Extended Value iteration (EVI  \cite{auer2008near}) }\label{alg:EVI_cap}
\begin{algorithmic}
\State $\bm{Initialize}$: State values $\tilde{h}_{0}(s)=0$ for all $s \in\mathcal{S}$.
\For{ iteration $z=0,1, 2,3 \cdots$}\\ 
 For each $s \in\mathcal{S}$, compute value iteration :\\
 $\tilde{h}_{z+1}(s)=\min_{a \in \mathcal{A}} \tilde{\Psi}_{z}(s,a) $, where\\
 $\tilde{\Psi}_{z}(s,a)= c(s,a)+ \min_{w \in \mathcal{H}_{w}(s,a)} \{ \sum_{s' \in \mathcal{S}} \tilde{h}_{z}(s') w(s'|s,a)\}  $. \\
 The inner minimization calculation can be found by using a similar technique as in \cite{auer2008near}.\\ 
Find a policy $\tilde{\pi}$, such that, $\tilde{\pi}(s)=\argmin_{a \in \mathcal{A}} \tilde{\Psi}_{i}(s,a)$.\\
\If{$\max_{s \in \mathcal{S}}  \{ \tilde{h}_{z+1}(s)-\tilde{h}_{z}(s)\}$$-$$\min_{s \in \mathcal{S}} \{ \tilde{h}_{z+1}(s)-\tilde{h}_{z}(s)\}$$ \leq \epsilon$}\\
Break the \textbf{for} loop.\\
\EndIf
\EndFor\\
Return a policy $\tilde{\pi}$.
\end{algorithmic}
\end{algorithm}

Similar to \cite{cheung2023nonstationary}, in AEC-SW-UCRL2 algorithm, we use a fixed window size which depends on the known variation budgets $V_{\lambda}$ and $V_{q}$ and is given by $W=4 |\mathcal{S}|^{2/3} |\mathcal{A}|^{1/2} T^{1/2}(V_{\lambda}+V_{q})^{-1/2}$; this is where the knowledge of $V_\lambda$ and $V_q$ is used, along with Lemma~\ref{lemma:var_boudget-relation}. However, we relax the assumption of known    $V_{\lambda}$ and $V_{q}$ in the AEC-BORL algorithm proposed in the next subsection, which tunes the window size for different blocks to adapt learning.   

\subsection{AEC-BORL algorithm
}\label{section:Age Aware BORL} 
 In this section, we adapt the BORL algorithm  \cite{cheung2023nonstationary} for RL in a non-stationary environment, in order to develop our  AEC-BORL algorithm for AoI minimization in non-stationary EH system model with unknown variation budgets $V_{\lambda}$ and $V_q$. In AEC-BORL,  the entire time horizon is divided into $\lfloor \frac{T}{L} \rfloor$ blocks of equal length $L$ and a last block whose length might be less than or equal to $ L$. AEC-BORL, like BORL,  probabilistically selects a window size $W_{\theta}$ for block $\theta \in \lfloor \frac{T}{L} \rfloor $ from a finite set $\mathcal{Q}$. At the starting of each block $\theta$, the AEC-BORL algorithm first selects a window size $W_{\theta} \in \mathcal{Q}$ by using some master algorithm, it then restarts the AEC-SW-UCRL2 algorithm with the selected window size $W_{\theta}$ as a sub-routine to choose actions for block $\theta$. Afterward, the total cost of block $\theta$ is fed back to the master, and then the probabilities of selecting a window size are updated accordingly. Here,  we use the EXP3.P algorithm for multi-armed bandit against an adaptive adversary  \cite[Section 3.2]{bubeck2012regret} as the master algorithm to choose $W_{\theta}$. While in AEC-SW-UCRL2, the optimal $W$ was chosen as a function of the variation budgets, AEC-BORL has to learn the optimal $W_{\theta}$ using a bandit algorithm due to the lack of a prior knowledge of the variation budgets.  Note that, the authors in \cite{cheung2023nonstationary} have also used the EXP3.P algorithm for choosing window sizes in BORL for the reward maximization problem,  but we have modified and re-parameterized it for the AoI cost minimization problem under EH remote sensing. It is also worth pointing out that the threshold comparison step is executed in AEC-BORL because AEC-SW-UCRL2 is called as a subroutine in AEC-BORL.

\subsubsection{Outline of AEC-BORL}\label{parameters-AEC-BORL}
In order  to ensure a bounded single stage cost, we assume that once $K(t)$ crosses a very large pre-defined number $K_{max}$, the corresponding single stage cost for $r(t)=0$ is $K_{max}$, and the age component in the state is also saturated at $K_{max}$. 
For   AEC-BORL, we use the following parameters:
\begin{eqnarray}\label{BORL_parameter1}
&& L=  \bigg \lfloor 3  \bigg(\frac{B}{K_{max}}\bigg)^{\frac{2}{3}} A^{\frac{1}{2} } T^{\frac{1}{2}} \bigg \rfloor, \hspace{0cm} \Delta_{W}= \lfloor L \rfloor,  \hspace{0cm}  \Delta= \Delta_{W}+1, \nonumber\\
&& \hspace{2.5cm} \mathcal{Q}=\{L^{0}, \lfloor L^{\frac{1}{\Delta_{W}} } \rfloor , \cdots, L\}
\end{eqnarray}

Here $L$ is the block length and set $\mathcal{Q}$ consists of all possible window sizes. Also, $\tilde{C}_{\theta}(W_{\theta},s)$ denotes the total (random) cost for running the AEC-SW-UCRL2 algorithm with selected window size $W_{\theta}$ for block $\theta$,  starting from an initial state $s$. 
Considering each element of set $\mathcal{Q}$ as an arm, the EXP3.P algorithm for choosing window size can be explained by following steps:

\begin{enumerate}
  \item Initialize,
  \begin{eqnarray} \label{BORL_parameter2}
&&\alpha=0.95 \sqrt{\frac{\ln \Delta}{\Delta \lceil T/L \rceil}}, \hspace{0.5cm} \beta= \sqrt{\frac{\ln \Delta}{\Delta \lceil T/L \rceil}}, \nonumber\\
&& \hspace{1.8 cm}\gamma= 1.05 \sqrt{\frac{\Delta \ln \Delta}{ \lceil T/L \rceil}}
\end{eqnarray}

 where $\alpha > 0$ and $\beta, \gamma \in [0,1]$.
 
\item Initially at time $t=1$, we set $g_{l,1}=0$ for all $l \in \mathcal{M}$ where $ \mathcal{M}=  \{0, 1, \cdots \Delta_{W}\}$.
\item The AEC-BORL algorithm first observes the state $s((\theta-1)L+1)$ at the starting of each block $\theta \in \lceil \frac{T}{L} \rceil$ and then computes the probability distribution:
\begin{eqnarray} \label{BORL_parameter3}
\chi_{l,\theta}= (1-\gamma)\frac{\exp{(\alpha g_{l,\theta}})}{\sum_{l^{'} \in M}  \exp{ (\alpha g_{l^{'},\theta}})} + \frac{\gamma}{\Delta},   \hspace{0.5 cm} \forall l
\end{eqnarray}

For block $\theta$, it selects $l_{\theta}=l$ with probability $\chi_{l,\theta}$ which gives the window size $W_{\theta}= \lfloor L^{l_{\theta}/\Delta_{W}} \rfloor $
\item After that, for each round ~$t$ in block $\theta$, the AEC-BORL algorithm with initial state $s((\theta-1)L+1)$, 
selects optimal actions by running the AEC-SW-UCRL2 algorithm with window size $W_{\theta}$.  At the end of each block $\theta$, the AEC-BORL algorithm observes the
total costs $\tilde{C}_{\theta}(W_{\theta}, s((\theta-1)L+1))$ and divides it by $LK_{max}$ to bring it to   $[0, 1]$ range, and update:
\begin{eqnarray} \label{BORL_parameter4}
g_{l,\theta+1}
&=& g_{l,\theta}+(\beta+ \mathbf{1}_{\{l_{\theta}=l\}} \cdot  \tilde{C}_{\theta}(W_{\theta}, s((\theta-1)L+\nonumber\\
&&1)))/ (LK_{max}\chi_{l,\theta}), \hspace{0.5cm} \cdots \forall l  
\end{eqnarray}
\end{enumerate}
\normalsize

The pseudocode of AEC-BORL algorithm is described by Algorithm \ref{alg:AEC_borl}.

 \begin{algorithm}
\caption{AEC-BORL Algorithm }\label{alg:AEC_borl}
\begin{algorithmic}
\State $\bm{Input}$: $\mathcal{S},\mathcal{A}, T$, initial state   $s(1) $.
\State $\bm{Initialize}$: $L, \Delta_{W}, \Delta , \mathcal{Q}$ by using \eqref{BORL_parameter1}  and $\alpha, \beta, \gamma$ by using \eqref{BORL_parameter2}.
\State $\bm{Set}$: $g_{l,1} \leftarrow 0$, $\forall l \in \mathcal{M}$, where $\mathcal{M} \leftarrow \{l^{'}: l^{'}\in \{0, 1, \cdots \Delta_{W}\}\}$.
\For{ $\text{block}$ $\theta=1, 2,3 \cdots \lceil \frac{T}{L} \rceil$} \\
Find probability distribution $\chi_{l,\theta}$ according to \eqref{BORL_parameter3} and select $l_{\theta}=l$ with probability $\chi_{l,\theta}$.
\State $\bm{Set}$: $W_{\theta} \leftarrow  \lfloor L^{l_{\theta}/\Delta_{W}} \rfloor$
\For{ $t=(\theta-1)L+1,  \cdots max (\theta L, T)$} \\
Use the AEC-SW-UCRL2 algorithm with window size $W_{\theta}$ for block $\theta$ to find optimal policy $\Bar{\pi}$ which depends on age, energy availability, and channel quality.\\ 
At the end of block $\theta$, observe the total cost $ \tilde{C}_{\theta}( W_{\theta}, s((\theta-1)L+1))$ and next state.
\EndFor\\
Update $g_{l,\theta+1}$ by using \eqref{BORL_parameter4}.
\EndFor
\end{algorithmic}
\end{algorithm}

\section{Multiple Sources ($N>1$) }\label{section: multi-source-nonstat}
First, we briefly discuss the results obtained from our previous work \cite{jaiswal2024whittlesindexbasedageofinformationminimization} for multiple source system working in a stationary environment ($\lambda_{i,t}=\lambda_{i}$ and $q_{j,i,t}=q_{j,i} \hspace{0.1cm}, \forall t \geq 1$) which will be used for designing the optimal source scheduling and sampling algorithms for multi-source system under a non-stationary environment.
\subsection{ Solution for stationary setting}
Here, our objective is to find a policy $\bm{\pi}$ that minimizes the expected average AoI across all sources averaged over time, while adhering to the source scheduling constraint for probing at each time:
\begin{align}\label{eqn:main-problem-ms}
 \min_{\bm{\pi}}  \lim_{T \rightarrow \infty}\frac{1}{TN} \sum_{t=1}^T \sum_{i=1}^N \mathbb{E}_{\bm{\pi}} (K_i(t))\nonumber\\
 s.t. \hspace{0.2 cm } \sum_{i=1}^N  b_{i}(t) \leq 1, \hspace{0.6 cm } \forall t \geq 1
\end{align}

In our prior work \cite{jaiswal2024whittlesindexbasedageofinformationminimization}, the problem in \eqref{eqn:main-problem-ms} was formulated as a long-run average cost MDP. Let $s_{i} \in \mathcal{S}_{i}$ represent the generic state of the $i$-th source at time $t$; given by $s_{i} =(E_{i}, K_{i})$. The available buffer energy of source $i$ at time $t$ is denoted by $E_{i}(t) \in \{0, 1,...,B_{i}\}$, whereas the AoI of source node $i$ at time $t$ is denoted by $K_{i}(t)$. The state space for source~$i$ is thus $\mathcal{S}_{i} \doteq \{0,1,\cdots,B_i\} \times \mathbb{Z}_+$ with $(E_i, K_i)$ as a generic state of source $i$. When the $i$-th source is chosen to probe the channel at any time $t$ (i.e., $b_i(t)=1$)), it encounters a generic intermediate state $v_{i} = (E_{i},K_{i},C(i, \cdot))$ with the corresponding intermediate state space $\mathcal{V}_{i}=\{0, 1,\cdots,B_i\}\times \mathbb{Z}_{+} \times \{C_{1}, C_{2},\cdots, C_{m}\}$, where channel state $C(i, \cdot)$ is  obtained by probing the $i$-th source node.  At time~$t$, the centralized scheduler computes $\{b_i(t)\}_{1 \leq i \leq N}$, such that that $\sum_{i=1}^N b_i(t)=1$. Afterwards, a source $i$ with $b_i(t)=1$ selects  action $a_i(t) \in \{0,1\}$.

 In \cite{jaiswal2024whittlesindexbasedageofinformationminimization}, we replaced the hard scheduling constraint in \eqref{eqn:main-problem-ms} at each time by a soft constraint on the time-averaged mean number of channel probes:
\begin{align}\label{eqn:main-problem_CMDP-ms}
 \min_{\bm{\pi}}  \lim_{T \rightarrow \infty}\frac{1}{TN} \sum_{t=1}^T \sum_{i=1}^N  \mathbb{E}_{\bm{\pi}} (c_{i}(s_{i}(t), b_{i}(t),a_{i}(t)))\nonumber\\
 s.t. \hspace{0.2 cm } \frac{1}{TN} \sum_{t=1}^T\sum_{i=1}^N \mathbb{E}_{\bm{\pi}} [ b_{i}(t)] \leq  \frac{1}{N} \hspace{0.3 cm } 
 \normalsize
\end{align}
 Next, we employed the Lagrangian relaxation technique by using a Lagrange multiplier $\hat{\mu}$, to transform the aforementioned CMDP into an unconstrained MDP that   minimizes the following Lagrangian:
\begin{eqnarray}\label{eqn:Lagrange_CMDP_policy-ms}
 \mathcal{L}(\bm{\pi},\hat{\mu})  &=& \lim_{T \rightarrow \infty}\frac{1}{TN} \sum_{t=1}^T \sum_{i=1}^N  ( \mathbb{E}_{\bm{\pi}}  [c_{i}(s_{i}(t), b_{i}(t),a_{i}(t)]+\nonumber\\ 
 &&\hat{\mu}  \mathbb{E}_{\bm{\pi}}  [b_{i}(t)] ) 
\end{eqnarray}
\normalsize

Let us define $\mathcal{L}^*(\hat{\mu})= \min_{\bm{\pi}} \mathcal{L}(\bm{\pi},\hat{\mu})$ as the Lagrangian dual function for any $\hat{\mu} \geq 0$, and let  $\bm{\pi}^*_{\hat{\mu}} \doteq\arg \min_{\bm{\pi}}\mathcal{L}(\bm{\pi},\hat{\mu})$. 

Remarkably, the task of finding an optimal policy $\bm{\pi}^*_{\hat{\mu}}$ for a given $\hat{\mu}$ may be divided into $N$ separate sub-problems across all sources. Hence, the Lagrangian in \eqref{eqn:Lagrange_CMDP_policy-ms} can be written as:
\begin{eqnarray}\label{eqn:Lagrange_alternate_CMDP_policy-ms}
 \mathcal{L}(\bm{\pi},\hat{\mu})  &=&  \frac{1}{N}  \sum_{i=1}^N \mathcal{L}_{i}(\bm{\pi}_i,\hat{\mu}) 
\end{eqnarray}
\normalsize

The Lagrangian for the $i$-th sub-problem of source $i$, denoted as $\mathcal{L}_{i}(\bm{\pi}_i,\hat{\mu})$, is defined as follows:
\begin{eqnarray}\label{eqn:Lagrange_persource_CMDP_policy}
  \mathcal{L}_{i}(\bm{\pi}_i,\hat{\mu}) &=& \lim_{T \rightarrow \infty}\frac{1}{T} \sum_{t=1}^T ( \mathbb{E}_{\bm{\pi}_i} [ c_{i}(s_{i}(t), b_{i}(t),a_{i}(t)]+\nonumber\\ 
 &&\hat{\mu}  \mathbb{E}_{\bm{\pi}_i}  [b_{i}(t)] ),  \hspace{0.6 cm}  \forall i \in \{1, 2, \cdots, N\}
\end{eqnarray}
\normalsize

Here, $\bm{\pi}_i$ refers to the policy employed for the $i$-th source. The sub-problem for the $i$-th source is an unconstrained average cost MDP. The solution of a $\alpha$-discounted cost MDP problem in the regime $\alpha \uparrow 1$ is used to solve the above-defined unconstrained average cost MDP. 

\subsubsection{Bellman Equation for the discounted cost problem \cite{jaiswal2024whittlesindexbasedageofinformationminimization} for source~$i$}
The optimal value function for state $(E_i,K_i)$ of source $i$ in the discounted cost problem is denoted by $J_{i}^{*}(E_i,K_i)$. The cost-to-go from an intermediate state $(E_i,K_i,C(i, \cdot))$ is represented by $W_{i}^*(E_i,K_i,C(i,\cdot))$. The Bellman equations for the $\alpha$-discounted cost MDP problem for each $i$-th source, as discussed in \cite{jaiswal2024whittlesindexbasedageofinformationminimization}, are given below: 
\begin{eqnarray} \label{eqn:Bellman-eqn-multiple-source-with-fading-general}
J^{*}_{i}(E_{i} \geq E_{s},K_i)
&=& \min \bigg\{\hat{u}_i(E_i,K_i), \hat{v}_i(E_i,K_i) \bigg\}\nonumber\\
\hat{u}_i(E_i,K_i) 
&=& K_i+\alpha \mathbb{E}_{A_i}J_i^{*}(\min\{E_i+A_i,B_i\},\nonumber\\
&&K_i+1),\nonumber\\
\hat{v}_{i}(E_i,K_i)
&=& \hat{\mu}+\sum_{j=1}^m q_{j,i} W^{*}_{i}(E_i,K_{i},C_{j})\nonumber\\ 
W^{*}_{i}(E_i,K_i,C(i, \cdot))
&=& \min\{K_i+\alpha \mathbb{E}_{A_i}J^{*}_{i}(\min\{E_i+A_i,\nonumber\\ 
&&B_i\},K_i+1),K_i(1-p(C(i, \cdot)))+\nonumber\\
&&\alpha p(C(i, \cdot))\mathbb{E}_{A_i}J^{*}_{i}(\min\{E_i-E_{s}+\nonumber\\
&&A_i,B_i\},1)+\alpha (1-p(C(i, \cdot)))\mathbb{E}_{A_i}\nonumber\\
&&J^{*}_{i}(\min\{E_i-E_{s}+A_i,B_i\},K_{i}+1)\}\nonumber\\
J^{*}_{i}(E_i< E_{s},K_i)
&=& K_i+\alpha \mathbb{E}_{A_i}J^{*}_{i}(\min\{E_i+A_i,B_i\},\nonumber\\
&&K_i+1)
\end{eqnarray}

The cost of not choosing the $i$-th source to probe channel state ($b_{i}(t)=0$) is represented by the term $\hat{u}_i(E_i, K_i)$ whereas the expected cost of choosing source $i$ to probe the channel state is expressed as $\hat{v}_{i}(E_i, K_i)$, which also includes the penalty cost $\hat{\mu}$ for selecting the $i$-th source to probe. The probed $i$-th source encounters an intermediate state $(E_i,K_i,C(i, \cdot))$ for which a sampling decision $a_{i}(t) \in \{0,1\}$ is chosen, which results in corresponding single stage cost and a random next state.  When $E_i<E_s$, $(b_{i}(t)=0, a_{i}(t)=0)$ is the only feasible action as provided in the last equation in \eqref{eqn:Bellman-eqn-multiple-source-with-fading-general}. 

For the $i$-th sub-problem with optimal policy $\pi_{i}^*$,  we define $\mathcal{I}(\hat{\mu})$ as the set of states where it is optimal to not probe the $i$-th source when the source probing charge is $\hat{\mu}$;  i.e., $ \mathcal{I}(\hat{\mu})= \{(E_i, K_i): \hat{u}_i(E_i, K_i)< \hat{v}_i(E_i, K_i)\}$. Furthermore, let us define $\mathcal{WI}_i(E_i, K_i)$ as the Whittle index for the state $(E_i, K_i)$ of source $i$, which represents the minimum charge $\hat{\mu}$ that would make probing and not probing decisions equally desirable in that state.

\begin{definition}
The sub-problem corresponding the $i$-th source is indexable, if $\mathcal{I}(\hat{\mu})$ increases monotonically form $\emptyset$ to the entire space as $\hat{\mu}$  increases from $0$ to $\infty$. The AoI minimization problem is indexable if all the $N$  sub-problems are indexable. 
\end{definition}

 Here, we assume that the indexability property holds for all the sub-problems.
At any time $t$, the optimal source scheduling and sampling policy for the $\alpha$-discounted  cost problem is to select the $i^*$-th source node with the highest Whittle's index to probe the channel state such that $E_{i^*} \geq E_s$, i.e., $ i^*(t)=\arg \max_{1 \leq i \leq N} \mathcal{WI}_i(t)$. Afterward, for the selected $i^*$-th source node, the most effective sampling strategy is a threshold policy based on $p(C (i^*,\cdot))$ which prescribes that node $i^*$ must be sampled if the probability $p(C(i^*,\cdot))$ is greater than or equal to the threshold $p_{th}(E_{i^*}, K_{i^*})$. 
For detailed analysis, please see our previous work \cite{jaiswal2024whittlesindexbasedageofinformationminimization}. 

\vspace{-6pt}

\subsection{WIT-SW-UCRL2 algorithm for non-stationary environment}\label{section: WIT Sliding Window UCRL2} 
In our prior work \cite{jaiswal2024whittlesindexbasedageofinformationminimization}, for stationary system setting with unknown channel states and EH characteristics, we 
proposed a Q-learning based algorithm named Q-WITS3 which seeks to learn Whittle’s indices and optimal policies. 
In this subsection, we assume that the energy arrival rate $\lambda_{i}$ and channel state probabilities $q_{j,i}$ are varying with time. Therefore, we divide the time horizon into episodes, and at the start of each episode, we estimate the energy arrival rate and channel state probabilities based on the recent history. By using these estimates, we learn Whittle's indices and optimal policies.

Here, we present a novel algorithm named Whittle's index and threshold based sliding window upper confidence bound reinforcement learning (WIT-SW-UCRL2) algorithm for a multi-source single sink wireless network working under a non-stationary environment  where the parameters $\lambda_{i,t}$ and $\{q_{j,i,t}\}_{1 \leq j \leq m}$ for each source $i \in \{1,2, \cdots,N\}$ change with time $t$. 
The tuple $(\mathcal{S}_i, \mathcal{A}_i, T, c_i, w_i)$ describes the MDP corresponding to the $i$-th source where $\mathcal{S}_i$,  
$\mathcal{A}_i= \{(0,0), (1,0), (1,1)\}$ and $c_i$ denote  the state space, the action space and the single stage cost function of the $i$-th source, respectively. The collection of state transition probabilities corresponding to the $i$-source for all time instants is denoted as $w_i = \{w_{i,t}\}_{t=1}^T$, where $w_{i,t}(\cdot|s_i, b_i, a_i)$ represents the probability distribution of the subsequent state of source $i$, given that the state-action pair at time $t$ is $(s_i, b_i, a_i)$.

 Let us denote by $s_{i} \in  \mathcal{S}_{i}$ a   generic state of the  $i$-th source,  which is given by  $s_{i} =(E_{i}, K_{i})$.
  At any time $t$, if the $i$-th source is selected to probe the channel (i.e., $b_i(t)=1$)), it encounters a generic intermediate state $ (E_{i},K_{i},C(i, \cdot))$. 
 The system state at any time $t$ is defined as $\bm{s}(t) = (s_{1}(t), . . . , s_{N}(t)) \in \mathcal{S}$ where $\mathcal{S} = \mathcal{S}_{1} \times  \cdots, \times \mathcal{S}_{N}$. 

Additionally, we define the variation budgets $V_{\lambda,i}$ and $V_{q,i}$  corresponding to the $i$-source as follows:
\begin{align}\label{eqn:main-problem1-var-ms}
V_{\lambda,i} \doteq  \sum_{t=1}^{T-1} V_{\lambda,i,t},  \hspace{1cm}   V_{q,i} \doteq \sum_{t=1}^{T-1} V_{q,i, t}
\end{align}
where $V_{\lambda,i,t} \doteq |\lambda_{i,t+1} - \lambda_{i,t}|$ and
$V_{q,i, t} \doteq max_{1 \leq j \leq m} |q_{j,i,t+1} - q_{j,i,t}|$. The learner is unaware of the values of $V_{\lambda,i,t}$ and $V_{q,i,t}$. Nevertheless, in this section, we assume that $V_{\lambda,i}$ and $V_{q,i}$ are known and are utilized to determine the sliding window size $W$. However, we will relax this assumption in the subsequent section.

We can also define the variation budget on state transition probabilities for each source $i$ as $V_{w,i}=\sum_{t=0}^{T-2} V_{w,i, t}$ 
where  $V_{w,i,t} \doteq max_{s_i \in \mathcal{S}_i, (b_i,a_i) \in \mathcal{A}_i}\lVert {w}_{i,t+1}(\cdot|s_i,b_i,a_i)- {w_{i,t}}(\cdot|s_i,b_i,a_i)\rVert_{1}$. The following lemma establishes a connection between $V_{w,i}$ and $V_{\lambda,i}$ and $V_{q,i}$ of each source $i$.  
\begin{lemma} \label{lemma:var_boudget-relation-ms}
For each source $i$ and every value $t$,  it holds that
$V_{w,i,t}\leq V_{\lambda,i,t+1}+V_{q,i,t+1}$ and consequently  $V_{w,i}\leq V_{\lambda,i}+V_{q,i}$.
\end{lemma}

\begin{proof}
The proof can be done by following similar steps as in Appendix ~\ref{appendix:proof-of-var_boudget-relation}.
\end{proof}
Our goal is to minimize the cumulative age averaged   across all the sources, subject to the constraint as in \eqref{eqn:main-problem-ms} :
\begin{eqnarray}\label{eqn:main-problem1-ms}
  \min_{\pi} \underbrace{\sum_{t=1}^T \sum_{i=1}^N  \frac{1}{N} ( \mathbb{E}_{\pi}[c_i(s_{i}(t),b_{i}(t), a_{i}(t))]}) \nonumber\\
   s.t. \hspace{0.2 cm } \sum_{i=1}^N  b_{i}(t) \leq 1, \hspace{0.6 cm } \forall t \geq 1
\end{eqnarray}


\subsubsection{Algorithm description}
WIT-SW-UCRL2 as described in Algorithm~\ref{alg:cap_ms} employs a sliding window of size $W$ to estimate the energy arrival rate and state transition probabilities of each source $i$ as well as to determine their respective confidence regions, accomplished by time averaging over a maximum of $W$ recent samples. 
The algorithm partitions the full time horizon $T$ into several episodes. The beginning and ending times of any episode $\rho$ are determined by $\tau(\rho)$ and $\tau(\rho+1)-1$, respectively, where $\tau(1)=1$.

We employ the empirical mean estimate of $ \lambda_{i,t}$ for source $i$ at any time step $t$ in episode $\rho$ as: 
\begin{eqnarray} \label{eqn:lambda_estimate_ms}
\hat{\lambda}_{i,t}&=& \frac{e_{i,t}}{n_t}
\end{eqnarray}
where  $e_{i,t} \doteq \sum_{h=max\{(\tau(\rho)-W),1\}}^{t} A_{i}(h)$,  and $n_t \doteq t-max\{(\tau(\rho)-W),1\}+ 1$ . Furthermore, the empirical mean estimator of $ q_{j,i,t}$ is given by:
\begin{eqnarray} \label{eqn:q_estimate_ms}
\hat{q}_{j,i,t}&=& \frac{o_{j,i,t}}{\tilde{n}_{i,t}}
\end{eqnarray}
where  $o_{j,i,t} \doteq \sum_{h=max\{(\tau(\rho)-W),1\}}^{t}X_{j,i}(h) \mathbf{1}\{b_{i}(h)=1\}$ with $X_{j,i}(h) \mathbf{1}\{b_{i}(h)=1\}$ denoting the indicator that the $i$-th source is selected to probe the channel and it encounters a channel state $j$ at time $h$ and $\tilde{n}_{i,t}$ is the number of times source $i$ is probed in the window size $W$ prior to time $t$, i.e., $\tilde{n}_{i,t} \doteq t-((max\{(\tau(\rho)-W),1\}+ 1) \mathbf{1}\{b_{i}(h)=1\})$. 

We also define the following counters for every state-action pair $(s_i, b_i, a_i)$ of source $i$ and every time $t$ in episode $\rho$.
\begin{eqnarray} \label{eqn:state-action-count_ms}
N_{i,t}(s_i,b_i,a_i)
&=& \sum_{h=max\{(\tau(\rho)-W),1\}}^{t-1} \mathbf{1}\{(s_{i}(h),b_{i}(h),\nonumber\\
&&a_{i}(h))=(s_i,b_i,a_i)\}\nonumber\\
N_{i,t}^+(s_i,b_i,a_i)
&=&max\{N_{i,t}(s_i,b_i,a_i),1\}
\end{eqnarray}
\normalsize

Also,  the counter for the number of times we come across state action pair $(s_{i}(t),b_{i}(t), a_{i}(t))=(s_i,b_i, a_i)$ for source $i$ in episode $\rho$ up to time $t$ is given as follows:
\begin{eqnarray} \label{eqn:state-action-count-within-episode_ms}
f_{i,\tau(\rho)}(s_i,b_i,a_i)&=& \sum_{h=\tau(\rho)}^{t-1} \mathbf{1}\{(s_{i}(h),b_{i}(h),a_{i}(h))\nonumber\\
&&=(s_i,b_i,a_i)\}
\end{eqnarray}
\normalsize

Here, the $\rho^{th}$ episode is terminated by the DM if either a time index $t$ is a multiple of $W$ or if $\max_{1 \leq i \leq N} (f_{i,\tau(\rho)}(s_{i}(t),b_{i}(t),a_{i}(t))$ $-$  $N^{+}_{i,\tau(\rho)}(s_{i}(t),b_{i}(t),a_{i}(t)) ) $ $> 0 $.

Next, the transition probabilities for each source $i$ are estimated using the following empirical mean estimator: 
\begin{eqnarray} \label{eqn:dist_estimate_ms}
&&\hat{w}_{i,t}(s_{i}'|s_i,b_i,a_i)\nonumber\\
&=& \frac{1}{N_{i,t}(s_i,b_i,a_i)^{+}} \bigg(\sum_{h=max\{(\tau(\rho)-W),1\}}^{t-1} \mathbf{1}\{s_{i}(h)=s_i,  \nonumber\\ 
&&b_{i}(h)=b_i,a_{i}(h)=a_i, s_{i}(h+1)=s_i'\}\bigg) 
\end{eqnarray}
\normalsize

and its corresponding confidence regions are given as follows: 
\begin{eqnarray} \label{eqn:confidence-region_ms}
&&\mathcal{H}_{w,i,t}(s_i,b_i,a_i)\nonumber\\
&\doteq& \bigg\{{w_i}:\lVert {w_i}(\cdot|s_i,b_i,a_i)-\hat{w}_{i,t}(\cdot|s_i,b_i,a_i)\rVert_{1} \nonumber\\
 & \leq& \sqrt{\frac{14S_i \log(2A_i\tau(\rho) / \delta)}{N_{i,\tau(\rho)}(s_i,b_i,a_i)^{+}}}\bigg\} 
\end{eqnarray}
\normalsize

Here,   $\delta$ is the confidence parameter and $confr_{w,i,\tau(\rho)}(s_i,b_i,a_i)=\sqrt{\frac{14S_i \log(2A_i\tau(\rho) / \delta)}{N_{i,\tau(\rho)}(s_i,b_i,a_i)^{+}}}$ is the confidence radius for source $i$.

In the WIT-SW-UCRL2 algorithm, the DM adheres to a stationary policy $\Bar{\pi}_{\tau(\rho)}$ throughout the episode $\rho$, and the optimal source scheduling and sampling decision is determined based on Whittle's index and the threshold values  of the sources. At each time $t$ during episode $\rho$, the WIT-SW-UCRL2 algorithm first performs the Whittle's index and threshold comparison steps by selecting the $i$-th source having the highest Whittle's index. Importantly, the Whittle's index of any source $i$ is computed at the beginning of an episode~$\rho$  by solving an MDP under the empirical mean estimators $\hat{\lambda}_{i,\tau(\rho)}$ and $\{\hat{q}_{j,i,\tau(\rho)}\}_{1 \leq j \leq m}$ available  at the beginning of the episode~$\rho$.  If $b_i(t)=1$, the algorithm further checks if  $K_{i}(t) \geq Thres_{i}(t)$, where the computation of the threshold $Thres_{i}(t)$ is done   based on $E_i(t)$ and $C(i,t)$. If $K_{i}(t) \geq Thres_{i}(t)$, the algorithm chooses  $a_{i}(t)=1$ provided that $E_{i}(t) \geq E_s$. On the other hand, if $K_{i}(t) < Thres_{i}(t)$, then for the $i$-th source the sampling policy $\tilde{\pi}_{\tau(\rho)}$  is determined by using EVI as long as $E_{i}(t) \geq E_s$. The algorithm utilizes the prior history over a window length $W$ starting from $\max\{\tau(\rho)-W,  1\}$ to  $(\tau(\rho)-1)$ to find the confidence regions $\mathcal{H}_{w,i, \tau(\rho)}(s_i,a_i)$ for all $(s_i,a_i)$ pairs of the source $i$. Here  $p(i,t) $, confidence regions $\mathcal{H}_{w, i,\tau(\rho)}(s_i,b_i,a_i)$ for all $(s_i,b_i,a_i)$ pairs of source $i$, and precision parameter $\epsilon=1 / \sqrt{ \tau(\rho)}$) are used as input parameters for EVI and can be found by following similar steps as in Algorithm~\ref{alg:EVI_cap}.

\begin{algorithm}
\caption{Whittle's Index and Threshold based Sliding Window  UCRL2 (WIT-SW-UCRL2) Algorithm }\label{alg:cap_ms}
\begin{algorithmic}
\State $\bm{Input}$: $\mathcal{S}_i,\mathcal{A}_i, K_i,B_i, T, N,W,$ $ \hspace{0.1 cm}\forall 1 \leq i \leq N$.
\State $\bm{Initialize}$: $\hat{\lambda}_{i,0}=0$, $\hat{q}_{j,i,0}=\frac{1}{m}$,   $t\leftarrow 1$,  initial state   $s_{i}(1), \hspace{0.1 cm} \forall i,j$.
\For{ $\text{episodes}$ $ \rho=1, 2,3 \cdots$} 
\State $\bm{Set}$: $\tau(\rho) \leftarrow t$, $f_{i,\tau(\rho)}(s_i,b_i,a_i)  \leftarrow 0$, $\forall (s_i,b_i,a_i)$ of each $i$. 
\State $\bm{Compute}$: $\hat{\lambda}_{i,\tau(\rho) }$, $\hat{q}_{j,i,\tau(\rho) }$ and $N_{i,\tau(\rho)}(s_i,b_i,a_i)^{+}$ for all $(s_i,b_i,a_i)$ of each source $i$ by using \eqref{eqn:lambda_estimate_ms}, \eqref{eqn:q_estimate_ms} and \eqref{eqn:state-action-count_ms}, respectively.
\State $\bm{Compute}$: $\mathcal{WI}_i$,  for all  source $i$ by running VI for the MDP with parameters $\hat{\lambda}_{i,\tau(\rho) }$, $\hat{q}_{j,i,\tau(\rho) }$.
\While{ $t$ is not multiple of W and $\max_{1 \leq i \leq N} (f_{i,\tau(\rho)}(s_{i}(t),b_{i}(t),a_{i}(t))$ $-$  $N^{+}_{i,\tau(\rho)}(s_{i}(t),b_{i}(t),a_{i}(t)) ) \leq 0 $}
\State Select $ i^*= \arg \max_{i} \mathcal{WI}_i$ such that $E_{i^*}(t)\geq E_s$.
Observe the current channel state $C(i^*,t)$ and consequently  $p(i^*,t)$.
\State $\bm{Compute}$: $\alpha_{i^*}(\tau(\rho))= e_{i^*,\tau(\rho)-1}+1,$ 
\State $\beta_{i^*}(\tau(\rho))= n_{\tau(\rho)-1 }-e_{i^*,\tau(\rho)-1} +1 $.
\State  $Thres_{i^*}(t)$  $= max\big(\frac{\alpha_{i^*}(\tau(\rho))+\beta_{i^*}(\tau(\rho))}{E_{i^*}(t)\alpha_{i^*}(\tau(\rho))}, 1/p(i^*,t)\big) $, 
\If{ Age $K_{i^*}(t) \geq Thres_{i^*}(t) $} \\
Choose $a_{i{^*}}(t)=1$,  observe age cost $c_{i^*}(s_{i{^*}}(t),b_{i{^*}}(t),a_{i{^*}}(t))$ and the next state $s_{i{^*}}(t+1)$.\\
Update: $f_{i^*,\tau(\rho)}(s_{i{^*}}(t),b_{i{^*}}(t),a_{i^*}(t))  \leftarrow f_{i^*,\tau(\rho)}(s_{i^*}(t),b_{i^*}(t),a_{i^*}(t))+1$, $ t \leftarrow t+1$.
\Else{  Find a $ (1 / \sqrt{\tau(\rho)})$ optimal source sampling policy $\tilde{\pi}_{\tau(\rho)}$ by using EVI algorithm using confidence region $\mathcal{H}_{w,i{^*},\tau(\rho)}(s_{i^{*}},b_{i^{*}},a_{i^{*}})$, $\forall$ $(s_{i{^*}},b_{i{^*}},a_{i{^*}})$ computed by \eqref{eqn:confidence-region_ms}:} \\
EVI$(\mathcal{H}_{w,i{^*},\tau(\rho)}, p(i^*,t);1 / \sqrt{ \tau(\rho)})$ $\rightarrow$ $(\tilde{\pi}_{\tau(\rho)})$. \\
Choose $a_{i^{*}}(t)=\tilde{\pi}_{\tau(\rho)}(s_{i^{*}}(t))$,  observe age cost $c_{i^*}(s_{i^*}(t),b_{i^*}(t),a_{i{^*}}(t))$ and the next state  $s_{i{^*}}(t+1)$.\\ 
Update: $f_{i^*,\tau(\rho)}(s_{i^*}(t),b_{i^*}(t),a_{i^*}(t))  \leftarrow f_{i^*,\tau(\rho)}(s_{i^*}(t),b_{i^*}(t),a_{i^*}(t))+1$, $ t \leftarrow t+1$.
\EndIf
\State  for any  $i \neq i^*$ choose $ b_{i}(t)=0$ and $a_{i}(t)=0$ , 
 observe age cost $c_{i}(s_{i}(t),b_{i}(t),a_{i}(t))$ and $s_{i}(t+1)$.\\
Update: $f_{i,\tau(\rho)}(s_{i}(t),b_{i}(t),a_{i}(t))  \leftarrow f_{i,\tau(\rho)}(s_{i}(t),b_{i}(t),a_{i}(t))+1$, $ t \leftarrow t+1$.
\If{$t>T$}
\\The algorithm is terminated. 
\EndIf
\EndWhile
\EndFor
\end{algorithmic}  
\end{algorithm}

The WIT-SW-UCRL2 algorithm uses a fixed window size  given by $W= \max_i 4 |\mathcal{S}_i|^{2/3} |\mathcal{A}_i|^{1/2} T^{1/2}(V_{\lambda,i}+V_{q,i})^{-1/2}$, which depends on the variation budgets.  Nevertheless, in the WIT-BORL algorithm presented in the next subsection, we alleviate this restriction. 

\subsection{WIT-BORL algorithm}\label{section: WIT BORL}
Here we present our WIT-BORL algorithm for minimizing the average age across all the sources in a non-stationary multiple EH source system model with unknown variation budgets $V_{\lambda, i}$ and $V_{q,i}$, $\forall i \in \{1,2,\cdots,N\}$. This algorithm adjusts the window size for various blocks to accommodate learning with the assumption that the maximum age of any source $i$ is bounded above by $K_{i,max}$. The entire time horizon is partitioned into $\lfloor \frac{T}{L} \rfloor$ blocks of uniform duration $L$, with the possibility that the length of the last block is less than or equal to $L$. The window selection for WIT-BORL is similar to the window selection of AEC-BORL which uses a probabilistic method to choose a window size $W_{\theta}$ per block $\theta \in \lfloor \frac{T}{L} \rfloor $ from a finite collection $\mathcal{Q}$. At the beginning of each block $\theta$, the WIT-BORL algorithm initially chooses a window size $W_{\theta} \in \mathcal{Q}$ using an EXP3.P master algorithm as in Section~\ref{section:Age Aware BORL}. It then executes the WIT-SW-UCRL2 algorithm with the selected window size $W_{\theta}$ as a sub-routine to choose actions for block $\theta$. The overall cost across all the sources of block $\theta$ is fed back to the master, and subsequently, the probabilities of selecting a window size are adjusted correspondingly. The Whittle's index and threshold comparison steps are done in WIT-BORL for source probing and source sampling, respectively because WIT-SW-UCRL2 is called as a subroutine in WIT-BORL.

The parameters used in the WIT-BORL are the same as in the AEC-BORL algorithm (Section \ref{parameters-AEC-BORL}) except that $B$, $K_{max}$,  and $A$ in \eqref{BORL_parameter1} are replaced by  $\max_i B_i$, $\max_i K_{i,max}$,  and $A_i$ (since the action space is same for all $i$)  in the calculation of $L$. Also, the steps involved in the EXP3.P algorithm of WIT-BORL for choosing window size are similar to that of  EXP3.P algorithm of AEC-BORL (Section \ref{parameters-AEC-BORL}). The only difference is that the WIT-BORL algorithm at the beginning of each block $\theta$ observes the state $s_{i}((\theta-1)L+1),\forall i$ (since we have multiple sources), and calculates the probability distribution by using \eqref{BORL_parameter3}  which determines the window size $W_{\theta}$.
Next, the WIT-BORL algorithm with starting states $s_{i}((\theta-1)L+1), \forall i$ chooses the best course of action for each round ~$t$ in block $\theta$ by executing the WIT-SW-UCRL2 algorithm with window size $W_{\theta}$, and  monitors the total costs $\sum _{i=1}^N  \tilde{C}_{i,\theta}(W_{\theta}, s_{i}((\theta-1)L+1))$ at the end of each block $\theta$, and then updates $g_{l,\theta+1}$ by \eqref{BORL_parameter4} using the replacement $  \tilde{C}_{\theta}(W_{\theta}, s((\theta-1)L+1))$ $\leftarrow$  $(\sum _{i=1}^N  \tilde{C}_{i,\theta}(W_{\theta}, s_{i}((\theta-1)L+1)))/N$. 

\vspace{-6pt}
\section{Numerical Results}\label{section: numerical_results}
\subsection{Stationary system model ($N=1$)}\label{subsection:numerical result Stationary system model}
We consider $m=5$ (number of channel states) and i.i.d. $Bernoulli(\lambda)$ energy arrival process. The channel state occurrence probabilities are given by $\bm{q}=[0.2,0.2,0.2,0.2, 0.2]$, with their respective packet success probabilities $\bm{p}=[0.9, 0.7, 0.5, 0.3, 0.1]$. We assume $B=9$, $K_{max}=100$ and  $E_{s}=1$ unit.  Numerical exploration reveals that the optimal source sampling policy is a threshold policy on either $K$ or $p(C)$, which validates Conjecture~\ref{conjecture:SSSP-with-fading-policy-structure}. It is observed from Figure ~\ref{fig_Thres}(a) that the threshold $K_{th}(E,C)$ decreases with $E$, $\lambda$ and $p(C)$,  since the EH node tries to sample a process more aggressively if there is higher energy available in the energy buffer or if the packet success probability is high. Also, Figure~\ref{fig_Thres}(b) shows that $p_{th}(E,K)$ decreases with $E$, $\lambda$ and $K$, since higher energy available in the buffer or higher age $K$ results in more aggressive sampling by the EH node.

\begin{figure}[htb]
  \begin{center}
 \subfloat[]{\includegraphics[height=4cm,width=9cm]{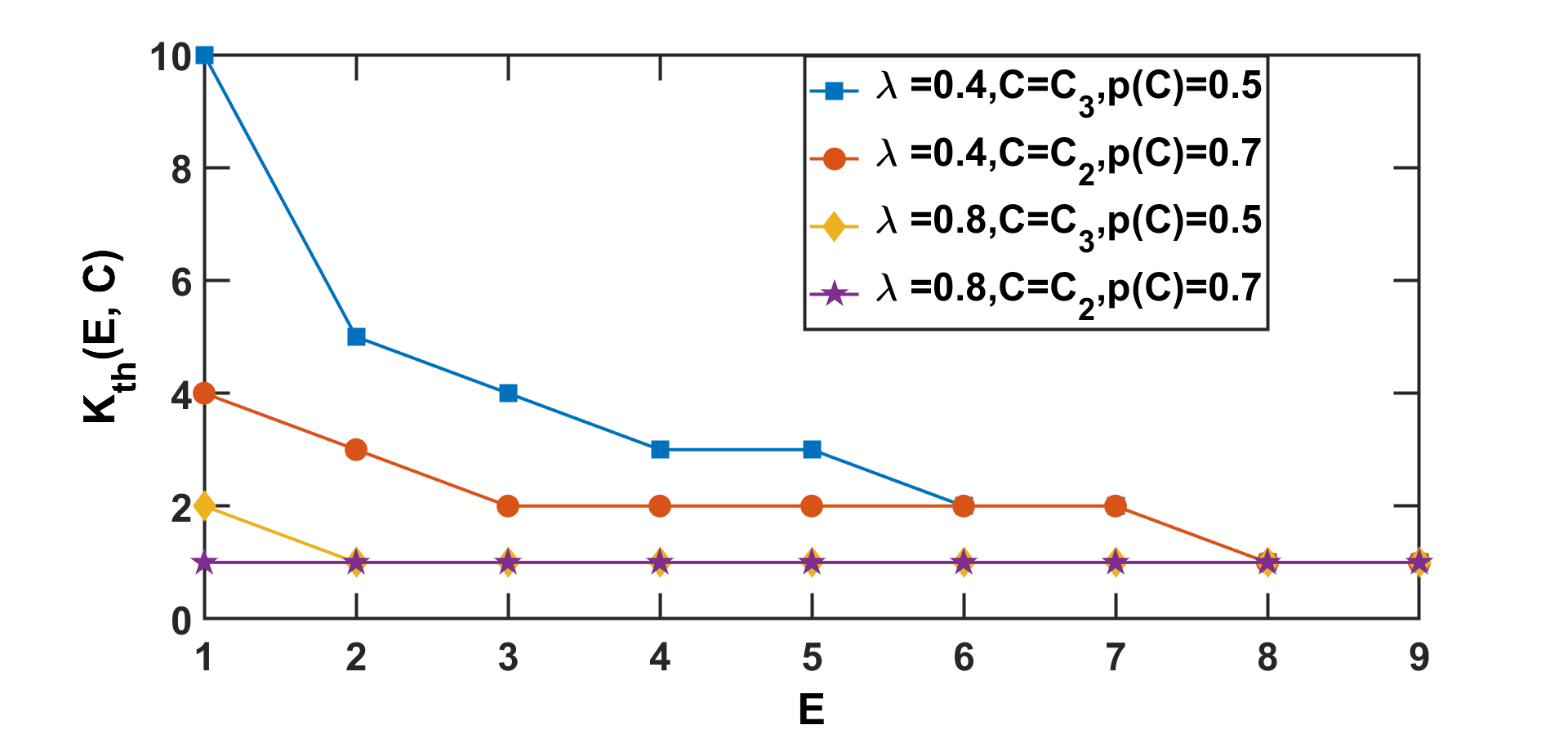}}
  \hfill
  \subfloat[]{\includegraphics[height=4cm,width=9cm]{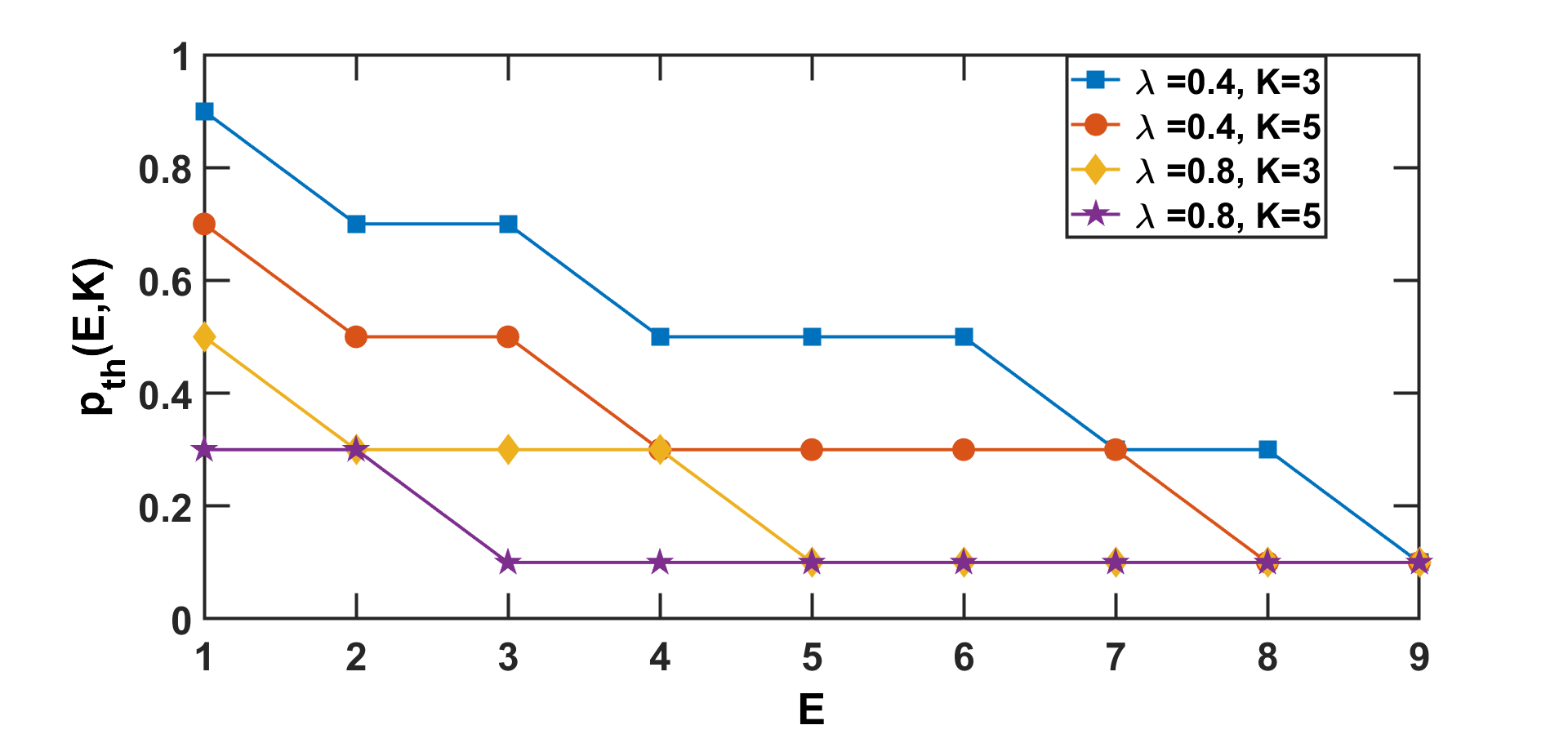}}
 \caption{Stationary system model: Variation of (a) $K_{th}(E, C)$ with $E$, $\lambda$, $p(C)$ and (b) $p_{th}(E,K)$ with $E, K, \lambda$.}
 \label{fig_Thres}
 \end{center}
 \vspace{-18pt}
 \end{figure}

\subsection{Non-Stationary system model ($N=1$)}\label{subsection:numerical result non Stationary system model}
We consider $m=2$, $\bm{q}_{t}=(q_{1,t}, q_{2,t})$, $\bm{p}=[0.8, 0.2]$ $B=7$, $E_{s}=1$ unit, $K_{max}=10$. The time varying  $\lambda_{t}$  and $\bm{q}_{t}$ are modelled by using following functions:
\begin{eqnarray} \label{lambda_t}
&&\lambda_{t}= (0.3+0.2\cos(2 \pi t/4)), \nonumber\\
&&q_{2,t}= (0.5+0.2\sin(2 \pi t/4)), \nonumber\\
&&  q_{1,t}=1-q_{2,t} 
\end{eqnarray}
We set the time horizon $T = 5000$, and compare the cumulative cost of the algorithms after averaging over ten sample paths. 
It is observed from Figure~\ref{fig_CumCost} that our proposed AEC-SW-UCRL2 and AEC-BORL algorithms outperform the standard  SWUCRL2 and BORL algorithms; this happens because, motivated by the theoretical findings in Section~\ref{section:single-sensor-single-process_stationary}, we have introduced one extra step of comparing the age against a threshold. We also observe that AEC-BORL  performs better than AEC-SW-UCRL2 since the former uses a varying window size.  

\begin{figure}[H]
  \begin{center}
 \includegraphics[height=4cm,width=9cm]{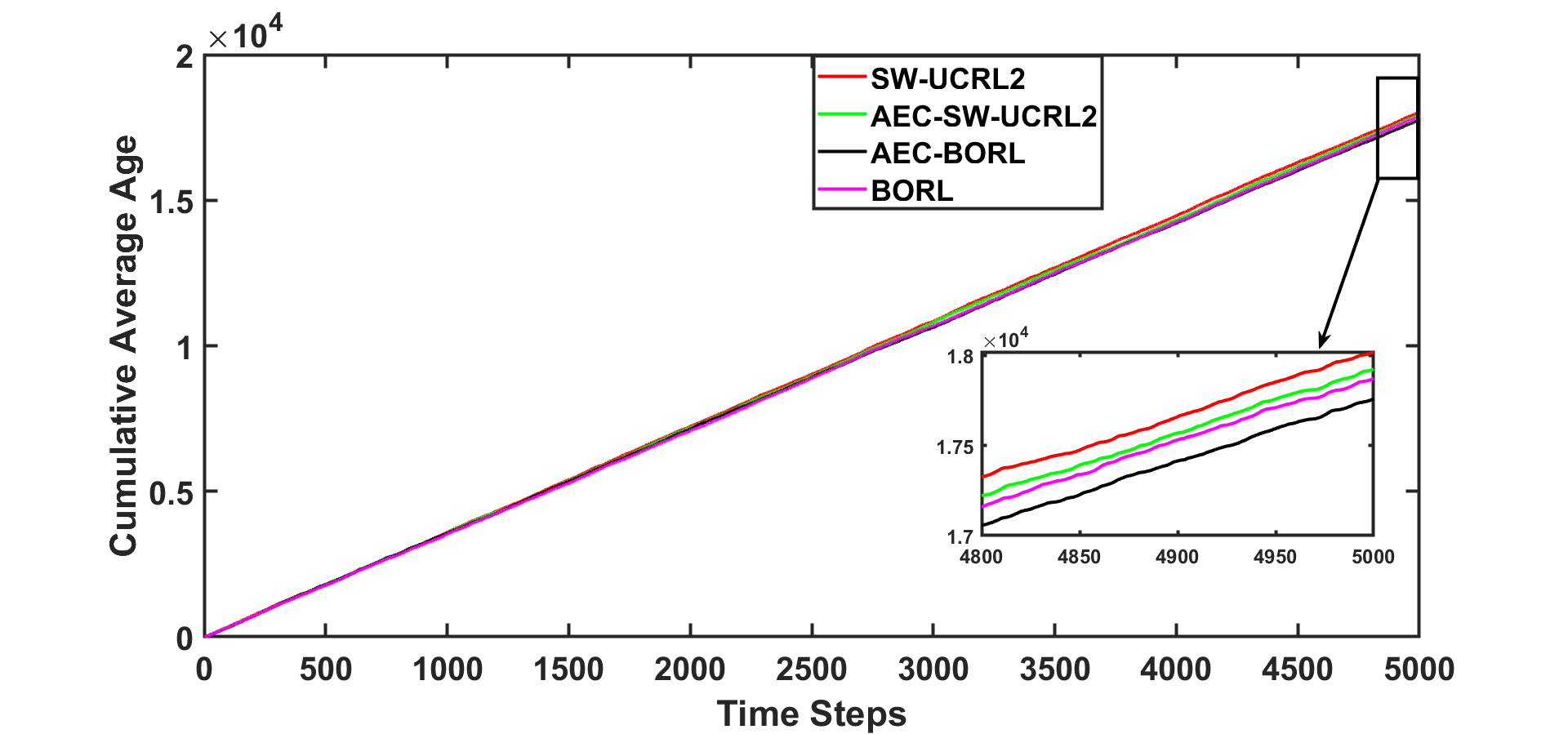}
 \caption{ $N=1$: Cumulative age  for non-stationary model.}
 \label{fig_CumCost}
 \end{center}
 \end{figure}
 
\subsection{Non-Stationary system model ($N>1$)}\label{subsection:numerical result non-Stationary multi-source system model}
Here, we consider $N=3$, $m=2$, $\bm{q}_{i,t}=(q_{1,i,t}, q_{2,i,t})$, $\forall i \in \{1,2,\cdots,N\}$, $\bm{p}=[0.8, 0.2]$, $B_i=7$, $E_{s}=1$ unit, $K_{i,max}=10$, $\forall i$. The time varying  $\lambda_{i,t}$  and $\bm{q}_{i,t}$ for each $i$ are modelled by using following functions:
\begin{eqnarray*}
&& \hspace{1.2 cm} \lambda_{1,t}= (0.3+0.2\cos(2 \pi t/4)),  \nonumber\\
&& \hspace{1.2 cm} \lambda_{2,t}= (0.5+0.2\cos(2 \pi t/4)), \nonumber\\
&& \hspace{1.2 cm} \lambda_{3,t}= (0.7+0.2\cos(2 \pi t/4)), \nonumber\\
&&q_{2,1,t}= (0.3+0.2\sin(2 \pi t/4)),\hspace{0.2 cm} q_{1,1,t}=1-q_{2,1,t}  \nonumber\\
&&q_{2,2,t}= (0.5+0.2\sin(2 \pi t/4)), \hspace{0.2 cm}q_{1,2,t}=1-q_{2,2,t} \nonumber\\
&&q_{2,3,t}= (0.7+0.2\sin(2 \pi t/4)),\hspace{0.2 cm} q_{1,3,t}=1-q_{2,3,t} \nonumber\\
\end{eqnarray*}
For multiple sources working in a non-stationary environment, we examine the  cumulative  age  averaged over   the sources and over multiple sample runs, by setting time horizon $T=5000$. We evaluate  the efficacy of the proposed algorithms by comparing them with the most relevant  baseline policies: (i) Maximum Age Sliding Window UCRL2 (MA-SW-UCRL2), and (ii) Random policy. The policy MA-SW-UCRL2 schedules the source having maximum age for probing the channel state and afterwards utilizes the SW-UCRL2 algorithm for sampling a selected source. The random policy selects a source for probing  uniformly at random. Due to lack of any competing algorithm for our system model, we compare the proposed algorithms with these baseline policies.    Figure~\ref{fig_CumCost-ms} demonstrates that our proposed algorithms WIT-SW-UCRL2 and WIT-BORL have lower cumulative average age as compared to the baselines MA-SW-UCRL2 and random policy. This happens due to the introduction of Whittle's index based source probing policy and  an extra step of age comparison with a threshold.   Next, we also observe that WIT-BORL outperforms WIT-SW-UCRL2 due to the utilization of a variable window size, despite the lack of knowledge of the variation budget.

 \begin{figure}[H]
  \begin{center}
 \includegraphics[height=4cm,width=9cm]{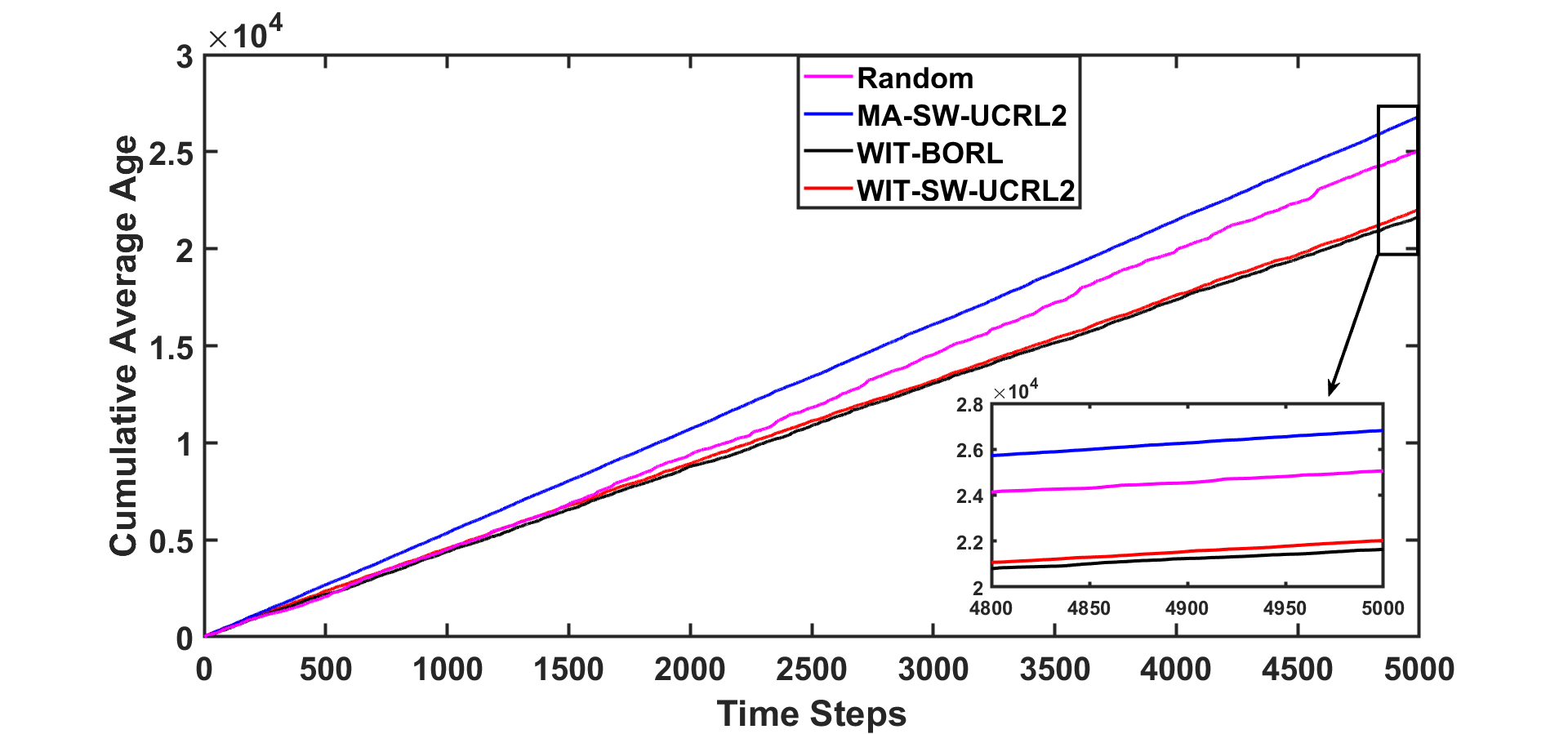}
 \caption{ $N=3$: Cumulative age for non-stationary model.}
 \label{fig_CumCost-ms}
 \end{center}
 \vspace{-6pt}
 \end{figure}
 
\section{Conclusion}\label{section:Conclusion}
In this paper, we have studied the time-averaged expected AoI minimization problem for single/multiple mobile EH source-assisted   remote sensing system. For a single source system working in a stationary environment, we have first shown that the optimal source sampling policy has a threshold structure. Thereafter, for the single source non-stationary system settings, we have proposed AEC-SW-UCRL2 and AEC-BORL algorithms, which are motivated from the literature, but exploit the threshold structure suggested by the results for the stationary setting. Furthermore, for multiple sources working in non-stationary environments, we have proposed WIT-SW-UCRL2 and WIT-BORL algorithms, which leverage Whittle's index based source scheduling policy derived from multiple sources working under a stationary system setting. Numerical results validate the theoretical results and demonstrate that our proposed algorithms outperform existing competing algorithms. In future, we seek to provide regret bounds for our proposed algorithms and also prove Theorem~\ref{theorem:SSSP-policy-K} without making Conjecture~\ref{conjecture:SSSP-with-fading-policy-structure}.

%
%

%
%
\bibliographystyle{IEEEtran}
\bibliography{ref_paper.bib}

\appendices
\section{Proof of Lemma~\ref{lemma:SSSP-J-increasing-in-K}}\label{appendix:proof-of-lemma-SSSP-J-increasing-in-K}
 We prove this result by value iteration: 
\begin{eqnarray}
&&J^{(t+1)}(E \geq E_{s},K, C)\nonumber\\ 
&=&min \bigg\{K+\alpha \mathbb{E}_{A, C^'}J^{(t)}(min\{E+A,B\},K+ 1, C^'),\nonumber\\
&&K(1-p(C))+\alpha p(C)\mathbb{E}_{A, C^'}J^{(t)}(min\{E-E_{s}+ A,\nonumber\\
&&B\},1, C^')+\alpha (1-p(C))\mathbb{E}_{A, C^'}J^{(t)}(min\{E-E_{s}+\nonumber\\
&&A,B\},K+1, C^')\}  \bigg\}\nonumber\\ 
&&J^{(t+1)}(E < E_{s},K, C)\nonumber\\ 
&=&K+\alpha \mathbb{E}_{A, C^'}J^{(t)}(min\{E+A,B\},K+1, C^')
\end{eqnarray}

Let us start with $J^{(0)}(s) = 0$ for all $s\in \mathcal{S}$. Clearly,
$J^{(1)}(E \geq E_{s}, K, C) =min \{K, K(1-p(C))\}$   and $J^{(1)}(E < E_{s}, K,C) = K$. Hence, for any given $E$ and $C$, the value function $J^{(1)}(E, K, C)$ is an increasing function of $K$. As an induction hypothesis, we assume that $J^{(t)}(E, K, C)$ is also an increasing function of $K$. Now, 
\begin{eqnarray}\label{value-iteration-ssns-with-fading-1}
&&J^{(t+1)}(E \geq E_{s},K, C)\nonumber\\
&=& min \bigg\{K+\alpha \mathbb{E}_{A, C^'}J^{(t)}(min\{E+A,B\},K+1, C^'),\nonumber\\
&& K(1-p(C))+\alpha p(C)\mathbb{E}_{A, C^'}J^{(t)}(min\{E-E_{s}+A, \nonumber\\
&&B\},1, C^')+\alpha (1-p(C))\mathbb{E}_{A, C^'}J^{(t)}(min\{E-E_{s}+ \nonumber\\
&&A,B\},K+1, C^')\}  \bigg\} 
\end{eqnarray}
\normalsize

We need to show that $J^{(t+1)}(E \geq E_{s},K,C)$ is also increasing in $K$. The first and second terms inside the minimization operation in  \eqref{value-iteration-ssns-with-fading-1} are increasing in $K$, by the induction hypothesis and from the fact that expectation is a linear operation. Thus, $J^{(t+1)}(E \geq E_{s},K,C)$ is  increasing in $K$. By similar arguments, we can claim that $J^{(t+1)}(E < E_{s} ,K,C)$ is increasing in $K$. Now, since $J^{(t)}(\cdot)\uparrow J^{*}(\cdot)$ as $t \uparrow \infty$ 
by \cite[ Proposition 2.1]{bertsekas2011dynamic}, $J^{*}(E,K,C)$ is also increasing in $K$. 

\section{Proof of Theorem~\ref{theorem:SSSP-policy-p}}\label{appendix:proof-of-theorem-SSSP-policy-p}

From \eqref{eqn:Bellman-eqn-SSSP_Probingt}, it is obvious that the optimal decision for $E \geq E_{s} $ is to sample the source if and only if the cost of sampling is lower than the cost of not sampling the source, i.e., 
$K+\alpha \mathbb{E}_{A,C^'}J^{*}(min\{E+A,B\},K+1,C^') \geq K(1-p(C))+\alpha \mathbb{E}_{A,C^'}J^{*}(min\{E-E_{s}+A,B\},K+1,C^')-
\alpha p(C)\bigg(\mathbb{E}_{A,C^'}J^{*}(min\{E-E_{s}+A,B\},K+1,C^')-\mathbb{E}_{A,C^'}J^{*}(min\{E-E_{s}+A,B\},1,C^')\bigg)$.
Now, by Lemma~\ref{lemma:SSSP-J-increasing-in-K}, $\mathbb{E}_{A,C^'}J^{*}(min\{E-E_{s}+A,B\},K+1,C^')- \mathbb{E}_{A,C^'}J^{*}(min\{E-E_{s}+A,B\},1,C^')$ is non-negative. Thus the  L.H.S. is independent of $p(C)$, whereas the R.H.S. is decreasing with $p(C)$. Hence, it is optimal to sample if  $p(C) \geq p_{th}(E,K)$ for some suitable threshold function $p_{th}(E,K)$.

   \section{Proof of Theorem~\ref{theorem:SSSP-policy-K}}\label{appendix:proof-of-theorem-SSSP-policy-K}

By using the similar argument as in the proof of Theorem ~\ref{theorem:SSSP-policy-p},  the optimal decision for $E \geq E_{s} $ is to sample the source if and only if the cost of sampling is lower than the cost of not sampling the source, i.e., 
$K+\alpha \mathbb{E}_{A,C^'}J^{*}(min\{E+A,B\},K+1,C^') \geq K(1-p(C))+\alpha \mathbb{E}_{A,C^'}J^{*}(min\{E-E_{s}+A,B\},K+1,C^')-
\alpha p(C)\bigg(\mathbb{E}_{A,C^'}J^{*}(min\{E-E_{s}+A,B\},K+1,C^')-\mathbb{E}_{A,C^'}J^{*}(min\{E-E_{s}+A,B\},1,C^')\bigg)$ which can be simplified to 
$\alpha \mathbb{E}_{A,C^'}J^{*}(min\{E+A,B\},K+1,C^') -\alpha \mathbb{E}_{A,C^'}J^{*}(min\{E-E_{s}+A,B\},K+1,C^')
\geq -Kp(C))-\alpha p(C)\bigg(\mathbb{E}_{A,C^'}J^{*}(min\{E-E_{s}+A,B\},K+1,C^')-\mathbb{E}_{A,C^'}J^{*}(min\{E-E_{s}+A,B\},1,C^')\bigg)$.

Now, by Conjecture \ref{conjecture:SSSP-with-fading-policy-structure}, the difference $\alpha \mathbb{E}_{A,C^'}J^{*}(min\{E+A,B\},K+1,C^') -\alpha \mathbb{E}_{A,C^'}J^{*}(min\{E-E_{s}+A,B\},K+1,C^')$ is increasing in $K$ and by Lemma~\ref{lemma:SSSP-J-increasing-in-K}, $\mathbb{E}_{A,C^'}J^{*}(min\{E-E_{s}+A,B\},K+1,C^')- \mathbb{E}_{A,C^'}J^{*}(min\{E-E_{s}+A,B\},1,C^')$ is non-negative. Thus the L.H.S. is increasing with $K$, whereas the R.H.S. is decreasing with  $K$. Hence, it is optimal to sample if $K \geq K_{th}(E,C)$ for some suitable threshold function $K_{th}(E,C)$.
\section{Proof of Lemma~\ref{lemma:var_boudget-relation}}\label{appendix:proof-of-var_boudget-relation}
The state transition probabilities  are given by:
\begin{equation}
\begin{aligned}
&w_{t}(s(t+1) \mid s(t)=(E(t)< E_{s}, K(t), C_i), a(t)=0) \\
&=
\begin{cases}
q_{j,t+1} \lambda_{t+1}, 
&\hspace{-1.2em} \begin{aligned}
&s(t+1) = \big(E(t+1)=\min\{E(t)\\
&\quad +1, B\}, K_{t+1}=\min\{K(t)+1, \\
&\quad  K_{\max}\}, C_j\big) \hspace{1.2em}\cdots  \bm{\text{Case } 1}
\end{aligned} \\\\
q_{j,t+1}(1 - \lambda_{t+1}), 
&\hspace{-1.2em} \begin{aligned}
&s(t+1) = \big(E(t+1)=E(t),\\
&\quad K_{t+1}=\min\{K(t)+1, K_{\max}\}, \\
&\quad  C_j\big)\hspace{4.4em}\cdots  \bm{\text{Case } 2}
\end{aligned} \\
0, & \text{otherwise} \hspace{2em} \cdots  \bm{\text{Case } 3}
\end{cases}
\end{aligned}
\end{equation}
\begin{equation}
\begin{aligned}
&w_{t}(s(t+1) \mid s(t)=(E(t)\geq E_{s}, K(t),C_i), a(t)=1) \\
&= 
\begin{cases}
q_{j,t+1}\lambda_{t+1}p_{j}, 
&\hspace{-1.2em} \begin{aligned}
&s(t+1)= \big(E(t+1)=\min\{\\
&\quad E(t)-E_s+1,B\},K_{t+1}= \\&\quad 1, C_j\big)\hspace{5em}\cdots  \bm{\text{Case } 4}
\end{aligned} \\\\
q_{j,t+1}(1-\lambda_{t+1})p_{j}, 
&\hspace{-1.2em} \begin{aligned}
&s(t+1)= \big(E(t+1)=E(t)-\\ 
&\quad E_s, K_{t+1}=1, C_j\big) \hspace{0.0em} \cdots  \bm{\text{Case } 5}
\end{aligned} \\\\
q_{j,t+1} \lambda_{t+1}(1-p_{j}),
&\hspace{-1.2em} \begin{aligned}
&s(t+1)=\big(E(t+1)=\min\{\\
&\quad E(t)-E_s+1,B\},K_{t+1}=\\& \quad \min\{K(t)+ 1, K_{\max}\}, C_j\big)\\&\quad 
\hspace{6.6em} \cdots  \bm{\text{Case } 6}
\end{aligned} \\\\
q_{j,t+1} (1-\lambda_{t+1})(1-p_{j}),
&\hspace{-1.2em} \begin{aligned}
&s(t+1)=\big(E(t+1)=E(t)-E_s,\\ 
&\quad K_{t+1}=\min\{K(t)+1,K_{\max}\},\\
&\quad  C_j\big) \hspace{4.7em}  \cdots  \bm{\text{Case } 7}
\end{aligned} \\
0, & \text{otherwise} \quad  \hspace{1.5em} \cdots  \bm{\text{Case } 8}
\end{cases}
\end{aligned}
\end{equation}

\normalsize

From the above two equations, we can write the increments in $w_t$ for various cases as follows:
\begin{eqnarray*}
&&\bigg|\delta w_{t} (s(t+1)|s(t)=(E(t)< E_{s}, K(t),C_i), a(t)=0)\bigg| \nonumber\\
&=&
 \begin{cases}
  \big|\delta q_{j,t+1}\big|\lambda_{t+1}+\big|\delta  \lambda_{t+1}\big|q_{j,t+1}   \hspace{1.7cm} \cdots \bm{\text{Case } 1} \\
 \big|\delta q_{j,t+1}\big|(1-\lambda_{t+1})+\big|\delta (1- \lambda_{t+1})\big| q_{j,t+1}   \cdots \bm{\text{Case } 2}\\
 0    \hspace{5.5cm}\cdots \bm{\text{Case } 3}
 \end{cases}
\end{eqnarray*}
\begin{eqnarray*}
&&\bigg|\delta w_{t} (s(t+1)|s(t)=(E(t)\geq E_{s}, K(t),C_i), a(t)=1)\bigg|\nonumber\\ 
&=&
 \begin{cases}
 \big|\delta q_{j,t+1}\big|\lambda_{t+1}p_{j}+ \big|\delta \lambda_{t+1}\big| q_{j,t+1}p_{j},  \hspace{1.5 cm}\cdots \bm{\text{Case } 4}  \\
 \big|\delta q_{j,t+1}\big|(1-\lambda_{t+1})p_{j}+ \big|\delta (1-\lambda_{t+1})\big| q_{j,t+1}p_{j}  \hspace{0cm} \cdots \bm{\text{Case } 5}  \\
  \big|\delta q_{j,t+1}\big|\lambda_{t+1}(1-p_{j})+ \big|\delta \lambda_{t+1}\big| q_{j,t+1}(1-p_{j}) \hspace{0cm}  \cdots \bm{\text{Case } 6}\\
 \big|\delta q_{j,t+1}\big|(1-\lambda_{t+1})(1-p_{j})+ \big|\delta (1-\lambda_{t+1})\big| q_{j,t+1}(1-p_{j})  \\
 \hspace{6.3cm} \cdots\bm{\text{Case } 7}  \\
 0  \hspace{6.1cm}\cdots \bm{\text{Case } 8}
 \end{cases}
\end{eqnarray*}
\normalsize

If $\lambda_{t+1} \in [\lambda_{min}, \lambda_{max}]$ and $ q_{j,t+1} \in  [q_{min},  q_{max}]$ $,  \forall j, \forall t$,  it follows that for every possible state action pair:
\begin{eqnarray*}
\big|\delta w_{t} (\cdot|s,a)\big| &\leq& \big|\delta q_{j,t+1}\big|\lambda_{max} + \big|\delta \lambda_{t+1}\big| q_{max} \\
&\leq &   \big|\delta q_{j,t+1}\big|+ \big|\delta \lambda_{t+1}\big| 
\end{eqnarray*}
Hence, $V_{w,t}\leq V_{\lambda,t+1}+V_{q,t+1}$, which implies    $V_{w}\leq V_{\lambda}+V_{q}$.


\end{document}